\documentclass[10pt,twocolumn,twoside]{IEEEtran} 

\usepackage{graphicx,amsmath,amsbsy,amssymb}
\usepackage{colortbl}	
\usepackage{tabularx}
\usepackage{theorem}	
\usepackage{algorithm,algorithmic}

\newcommand{\R}{{\mathbb R}}

\newcommand{\I}{{\rm I}}
\renewcommand{\P}{{\rm P}}
\newcommand{\LI}{L_{\rm I}}

\newcommand{\bLI}{\bar{L}_{\rm I}}
\newcommand{\bLP}{\bar{L}_{\rm P}}

\newcommand{\E}{{\cal E}}

\newcommand{\N}{{\mathcal{N}}}
\newcommand{\ra}{\rightarrow}
\newcommand{\la}{\leftarrow}

\newtheorem{theorem}{Theorem}
\newtheorem{assumption}{Assumption}
\newtheorem{lemma}{Lemma}
\newtheorem{remark}{Remark}

\newtheorem{definition}{Definition}

\title{Passivity-Based Distributed Optimization with Communication Delays Using PI Consensus Algorithm}

\author{Takeshi Hatanaka,~\IEEEmembership{Member,~IEEE}, Nikhil Chopra,~\IEEEmembership{Member,~IEEE},  
Takayuki Ishizaki,~\IEEEmembership{Member,~IEEE},\\ and Na Li,~\IEEEmembership{Member,~IEEE}
\thanks{T. Hatanaka (corresponding author) and T. Ishizaki
are with School of Engineering,
Tokyo Institute of Technology, 
2-12-1 S5-16 (Hatanaka)/W8-1 (Ishizaki), Ookayama, Meguro-ku, Tokyo 152-8550, JAPAN.
Tel: +81-3-5734-3316 (Hatanaka), +81-3-5734-2646, {\sf ishizaki@mei.titech.ac.jp} (Ishizaki).
Email: {\sf hatanaka@ctrl.titech.ac.jp} (Hatanaka), {\sf ishizaki@mei.titech.ac.jp} (Ishizaki)
N. Chopra is with Department of Mechanical Engineering, University of Maryland, 
College Park, MD 20742, USA.
Tel: +1-301-405-7011.
Email: {\sf nchopra@umd.edu}.
N. Li is with Electrical Engineering and Applied Mathematics of the School of Engineering and Applied Sciences,
Harvard University, 33 Oxford St, Cambridge, MA 02138, USA.
Tel: +1-617-496-1441.
Email: {\sf nali@seas.harvard.edu}.}
\thanks{
The results of this paper are partially presented in the authors' antecessor \cite{Nec},
but communication delays are not addressed there.}}
\begin{document}
\maketitle
\thispagestyle{empty}
\pagestyle{empty}

\begin{abstract}
In this paper, we address a class of distributed optimization problems in the presence of
inter-agent communication delays based on passivity.
We first focus on unconstrained distributed optimization and provide
a passivity-based perspective for 
distributed optimization algorithms. 
This perspective allows us to handle communication delays while using scattering transformation.
Moreover, we extend the results to constrained distributed optimization,
where it is shown that the problem is solved by just adding one more feedback loop
of a passive system to the solution of the unconstrained ones.
We also show that delays can be incorporated in the same way as the
unconstrained problems.
Finally, the algorithm is applied to a visual human localization problem using
a pedestrian detection algorithm.
\end{abstract}


%



\section{Introduction}

Passivity has been extensively studied for 
analysis and design of both linear and nonlinear systems with
broad applications in physical systems \cite{OL_BK}.
In the last decade, passivity has also been actively studied 
in control of network systems \cite{BAW_BK,HCFS_BK},
where inherent passivity preservation and energy dissipative property 
allow one to systematically 
analyze the network properties and design cooperative controllers.
The passivity-based approach is also employed in more advanced research fields 
including transportation networks \cite{KD_AT14}, building control \cite{MMW_CDC12}, 
power grids \cite{IB_CDC14}, cyber-security \cite{LC_TAC12},  
cyber-physical systems \cite{AG_EJC13} and human-swarm interactions \cite{HCF_CDC15}.

Recently, passivity has been also used in optimization and related application fields
\cite{arcak2}--\cite{GDOB_Pdh15}.
Wen and Arcak \cite{arcak2} address Internet congestion control,
and present a unifying passivity-based 
design framework for network utility maximization problems. 
The framework is applied to CDMA power control in \cite{arcak}.
Yamamoto and Tsumura \cite{YT_TR} apply a distributed algorithm called Uzawa's algorithm
to smart grid control based on passivity.
Relations between passivity-based cooperative control and network flow optimization
are also discussed in \cite{BZA_AT14}. 
B\"{u}rger and Persis \cite{BP_AT15} present a passivity-based solution to
distributed optimization with dynamic elements. 
Ishizaki et al. \cite{IUI_ECC16} prove equivalence between 
convex gradients and incrementally passive systems.
Bravo \cite{GDOB_Pdh15} presents a fully distributed algorithm, which does not require 
a subject to gather information from all agents.

All of the above papers address so-called resource allocation problems \cite{B_BK} or its variations.
In this paper, we deal with another type of
distributed optimization studied in 
\cite{NO_TAC09}--\cite{DE_ACC14}.
Nedi\'{c} and Ozdaglar \cite{NO_TAC09} present a distributed algorithm
which combines consensus algorithms and subgradient methods.
The results are extended to constrained problems in \cite{NOP_TAC10},
where a variation of the algorithm in \cite{NO_TAC09} is shown to 
ensure exact convergence to the optimal solution using a diminishing step size.
A solution to the problem with globally defined inequality and
equality constraints is presented by
Zhu and Martinez \cite{ZM_TAC12}.
The results of \cite{ZM_TAC12} are further extended to problems with partial knowledge
on the global constraints in \cite{CNS_TAC14}.
The proposed solution, termed primal-dual perturbed subgradient method,
is shown to outperform \cite{ZM_TAC12} in terms of the convergence speed.
Acceleration of the convergence speed is also addressed in \cite{GN_arxiv}. 
Extensions to random networks and a stochastic subgradient method
are found in \cite{LO_TAC10,RNV_JOTA}.

While the above works present discrete-time recursive processes to compute the solution,
Wang and Elia \cite{WE_AAC10,WE_CDC11} take a continuous-time algorithm and 
provide a control theoretic perspective for the distributed optimization algorithms, 
which are the most closely related to this paper. 
The solution in \cite{WE_AAC10,WE_CDC11}  is further extended to dynamic topologies by Droge and Egerstedt \cite{DE_ACC14}.

All of the solutions in \cite{NO_TAC09}--\cite{DE_ACC14} rely on
the inter-agent information exchanges, which is typically implemented using communication technology.
However, the problems caused by the communication have not been fully explored in the literature.
Among many of such problems, this paper treats communication delays inherent in communication,
which may slow down the convergence to wait for the arrivals of messages
or even destabilize the solution processes in the worst case.
To address the issue, we employ the notion of passivity. 

We start with presenting a passivity-based perspective for 
the algorithms based on the consensus algorithm and so-called PI (Proportional-Integral) consensus algorithm \cite{FYL_CDC06,BFL_CDC10}.
In particular, it is revealed that the PI  consensus-based solution is regarded as a feedback connection of passive systems.
Passivity-based formalism presented above allows one
to utilize rich knowledge established in the history of passivity-based control.

We next treat communication delays using our passivity-based perspective of the distributed algorithm.
Specifically, we show that the delays are successfully integrated with
the above solution 
by using the techniques in \cite{CS_CDC06} together with 
the scattering transformation \cite{HCFS_BK}.
Exact convergence to the optimal solution is then proved based on passivity.

The above results are then extended to a constrained optimization problem.
In this part, we start with a delay free case, and
show that the problem is solvable 
by just adding one more feedback loop of a passive system
originating from the gradient-based update of the Lagrange multiplier.
Note that the resulting architecture is similar to the one presented in \cite{WE_CDC11}.
However,  
\cite{WE_CDC11} achieves only convergence to an approximate solution in the presence of the nonlinear inequality constraints,
due to the use of barrier functions.
In contrast, our solution ensures exact convergence to the optimal solution
while avoiding the gain tuning of the barrier functions.
The results are also extended to the case with communication delays
by following the same procedure as the unconstrained problem.

Finally, the present algorithm is applied to a visual human localization problem using
a pedestrian detection algorithm.

The contribution of this paper is summarized as below:
The primary contribution is to present a distributed algorithm with robustness against
communication delays. The issue, delays in distributed optimization, 
is addressed in \cite{AD_CDC12}--\cite{HNE_A15}.
Fair comparison with these articles is difficult due to 
considerably different problem settings, but in general they do not 
prove exact convergence to the optimal solution whereas we do.
The secondary contribution is to reveal that the problem in \cite{NO_TAC09}--\cite{DE_ACC14}
can be treated within the passivity paradigm.
Thirdly, we handle general convex inequality constraints in this paper, 
while the other passivity-based approaches
\cite{arcak2}--\cite{GDOB_Pdh15} take only linear and/or scalar constraints.


\section{Preliminary}
\label{sec:0}

In this section, we introduce some terminologies used in this paper.

We first introduce  passivity.
Consider a system with a state-space representation
\begin{eqnarray}
\dot x = \phi(x,u),\ \ y = \varphi(x,u),
\label{eq:0.0}
\end{eqnarray}
where $x(t)\in \R^N$ is the state, 
$u(t) \in \R^p$ is the input and $y(t)\in \R^p$ is the output.
Then, passivity is defined as below.
\begin{definition}
The system (\ref{eq:0.0}) is said to be passive if there exists a 
positive semi-definite
function $S: \R^N \to \R_+ := [0, \infty)$, 
called storage function, such that 
\begin{eqnarray}
S(x(t)) - S(x(0)) \leq \int^{t}_0 y^T(\tau)u(\tau)d\tau
\label{eq:0.1}
\end{eqnarray}
holds for all inputs $u:[0, t]\to \R^p$, 
all initial states $x(0)\in \R^N$ and all $t \in \R^+$.
In the case of the static system $y = \varphi(u)$,
it is passive if $y^Tu = \varphi^T(u)u \geq 0$ for all $u\in \R^p$.
\end{definition}
As widely known, if $S$ is differentiable,
(\ref{eq:0.1}) can be replaced by
\begin{eqnarray}
\dot S(x(t)) \leq y^T(t)u(t).
\label{eq:0.2}
\end{eqnarray}
Passivity is known to be preserved for feedback interconnections of passive systems,
and closed-loop stability is also ensured 
under additional assumptions on strict energy dissipation and observability.
Please refer to \cite{OL_BK}--\cite{HCFS_BK} or other seminal books cited therein for more details on passivity.

We next introduce another notion closely related to passivity,
namely incremental passivity.
In the context of this paper, it is sufficient to define incremental passivity for a static system.
\begin{definition}
A static system $y = \varphi(u)$ is said to be incrementally passive if the function $\varphi$ satisfies
\begin{eqnarray}
(\varphi(u_1) - \varphi(u_2))^T(u_1 - u_2)\geq 0
\nonumber
\end{eqnarray}
for all $u_1\in \R^p$ and $u_2\in \R^p$.
\end{definition}

We also use the following fundamental tool in convex optimization.
\begin{definition}
\label{lem:4}
A function $f: \R^N \to \R$ is said to be convex
if the following inequality holds for any $x,y\in \R^N$.
\begin{eqnarray}
(\nabla f(x))^T(y-x) \leq f(y) - f(x)
\label{eq:0.4}
\end{eqnarray}
The function is said to be strictly convex if the inequality (\ref{eq:0.4}) strictly holds
whenever $x\neq y$.
\end{definition}
The following well known result links the convex functions and 
passivity theory.
\begin{lemma}\cite{B_BK}
\label{lem:0}
Consider a convex function $f: \R^N \to \R$.
Then, its gradient $\nabla f: \R^N \to \R^N$ is incrementally passive, i.e.,
the following inequality holds for any $x,y\in \R^N$.
\begin{eqnarray}
(\nabla f(x) - \nabla f(y))^T(x-y) \geq 0
\nonumber
\end{eqnarray}
If $f$ is strictly convex, the inequality strictly holds as long as $x\neq y$.
\end{lemma}

\section{Passivity-Based Perspective for Unconstrained Distributed Optimization}
\label{sec:2}

In this section, we consider the following distributed optimization problem
investigated in \cite{NO_TAC09}.
\begin{eqnarray}
\min_{z\in \R^N} f(z) := \sum_{i = 1}^n f_i(z)
\label{eq:2.1}
\end{eqnarray}
The function $f_i: \R^N \to \R\ (i=1,2,\dots, n)$ is 
the private cost function of agent $i$, which is assumed to be
unaccessible from agents other than $i$.
The subsequent discussions rely on the following assumption.
\begin{assumption}
\label{ass:1}
The functions $f_1, f_2, \dots, f_n$ are convex, continuously differentiable, and their gradients 
denoted by $\phi_i := \nabla f_i\ (i = 1,2,\dots, n)$ are locally Lipschitz, i.e., . for every point $x_0 \in \R^N$,
there exists its neighborhood ${\mathcal X}_0$ such that 
$\|\phi_i(x) - \phi_i(y)\| \leq L_0(x_0)\|x-y\|$ holds for all $x,y \in {\mathcal X}_0$
with some constant $L_0(x_0)$.
\end{assumption}

Throughout this paper, we assume that the set of the optimal solutions is not empty,
and the corresponding minimal value of $f$ is finite.
An optimal solution to (\ref{eq:2.1})
is denoted by $z^* \in \R^N$.
Since $f$ is also  continuously differentiable and convex under Assumption \ref{ass:1},
a vector $z^*\in \R^N$ is an optimal solution to the problem (\ref{eq:2.1})
if and only if the following equation holds \cite{B_BK}.
\begin{eqnarray}
\nabla f(z^*) = \sum_{i = 1}^n \phi_i(z^*) = 0
\label{eq:2.2}
\end{eqnarray}

\subsection{Consensus-Based Distributed Optimization}
\label{sec:2.1}

Suppose that each agent $i$ has an estimate of the optimal solution $z^*$,
denoted by $x_i \in \R^N$, and that $x_i$ is updated 
so that it converges to the set of optimal solutions.
The agents are assumed to be able to exchange information
with neighboring agents through a network modeled by a graph $G := (\{1,2,\dots, n\}, {\mathcal E})$
satisfying the following assumption.
\begin{assumption}
\label{ass:2}
The graph $G$ is undirected and connected.
\end{assumption}
The set of all neighbors of agent $i$ is denoted by $\N_i$.

\begin{figure}
\centering
\includegraphics[width=6cm]{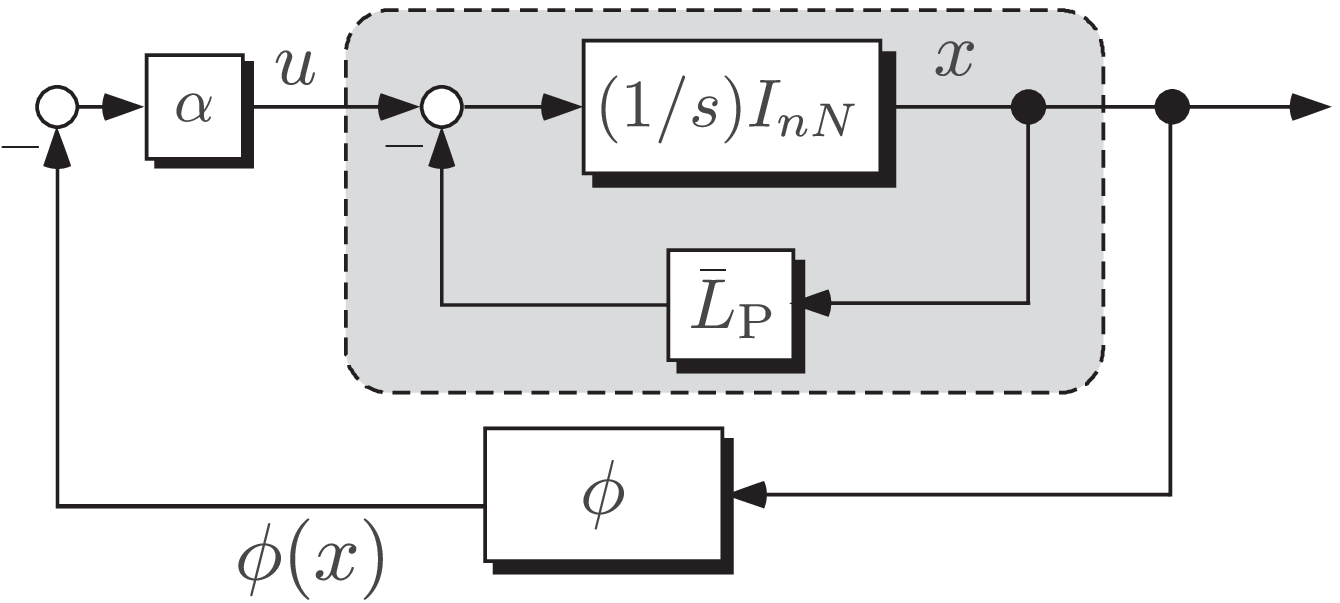}
\caption{Block diagram of the consensus-based distributed optimization algorithm for unconstrained optimization.
The system enclosed by the dashed line is passive from 
$u$ to $x$, and the bottom block $\phi$ is incrementally passive (Lemma \ref{lem:0}).}
\label{fig:1}
\end{figure}

Nedi\'{c} and Ozdaglar \cite{NO_TAC09} present an update rule of $x_i$ 
combining consensus algorithms and gradient descent algorithms whose 
continuous-time representation is given as
\begin{eqnarray}
\dot x_i = -\alpha \phi_i(x_i) + \sum_{j\in \N_i}a_{ij}(x_j - x_i),
\nonumber
\end{eqnarray}
where $\alpha$ is a positive scalar and $a_{ij}$ is the $(i,j)$-element of an 
adjacency matrix for $G$, where
$a_{ij} > 0$ if $(i,j) \in {\mathcal E}$
and $a_{ij} = 0$ otherwise, and
$a_{ij} = a_{ji}\ {\forall i,j}$.
The graph Laplacian associated with the adjacency matrix with elements $a_{ij}$ is denoted by $L_{\P}$.
The matrix $L_{\P}$ is symmetric and positive semidefinite
with a simple zero eigenvalue corresponding to the eigenvector $\textbf{1}_N$, where $\textbf{1}_N$ is the $N$ dimensional real vector whose elements all
equal to 1.

Collecting estimates $x_i$ as 
$x:=[x_1^T\ x_2^T\ \cdots\ x_n^T]^T$,
the evolution of $x$ is described by
\begin{eqnarray}
\dot x = -\alpha \phi(x) - \bar L_{\P} x,
\label{eq:2.4}
\end{eqnarray}
where 
$\phi(x):= [\phi_1^T(x_1)\
\cdots\
\phi_n^T(x_n)]^T$, 
$\bar L_{\P} := L_{\P} \otimes I_N$, the symbol $\otimes$ describes
the Kronecker product, and $I_N$ is the $N$-by-$N$ identity matrix.
The block diagram of the system is illustrated in Fig. \ref{fig:1}.

Let us now consider the system enclosed by the dashed line in Fig. \ref{fig:1}
whose input is denoted by $u$.
The dynamics of the system is described by
\begin{eqnarray}
\dot x = - \bar L_{\P} x + u.
\label{eq:11.4}
\end{eqnarray}
Take a storage function $S_{\rm P} := \frac{1}{2}\|x\|^2$. 
The time derivative of $S$ along the trajectories of (\ref{eq:11.4}) is then given as
\begin{eqnarray}
\dot S_{\rm P} = - x^T\bar L_{\P} x + x^Tu \leq x^Tu.
\label{eq:0.10}
\end{eqnarray}
We see from (\ref{eq:0.2}) that the system is passive.
Lemma \ref{lem:0} also ensures that the bottom block $\phi$ in Fig. \ref{fig:1}
is incrementally passive.
Thus, Fig. \ref{fig:1} is a feedback connection of a passive system
and an incrementally passive system.
However, because of the mismatch of the input-output pairs, 
the trajectories of $x_i$ do not converge to the solution as confirmed below.
The goal state is now formulated as
$x = x^*$ for $x^* 
:= \textbf{1}_n \otimes z^*$.
If $x = x^*$, the right-hand side of (\ref{eq:2.4})
is equal to
\begin{eqnarray}
-\alpha \phi(x^*) - (L_{\P} \textbf{1}_n) \otimes z^* 
= -\alpha\phi(x^*).
\nonumber
\end{eqnarray}
Although the sum of the elements $\phi_1(z^*),\phi_2(z^*),\dots, \phi_n(z^*),$ of $\phi(x^*)$ is zero from (\ref{eq:2.2}),
each element is not always zero, which implies that 
$x = x^*$ is not an equilibrium of (\ref{eq:2.4}).
Thus, the state trajectories do not converge to the goal state $x^*$.


\subsection{PI Consensus-Based Distributed Optimization}
\label{sec:2.2}

In this subsection, we focus on a distributed algorithm based 
on a PI consensus algorithm \cite{FYL_CDC06},
formulated as
\begin{eqnarray}
\dot x_i \!\!&\!\!=\!\!&\!\! \sum_{j\in \N_i}a_{ij}(x_j - x_i) - \sum_{j\in \N_i}b_{ij}(\xi_j - \xi_i)  + u_i,
\label{eq:2.9a}\\
\dot \xi_i \!\!&\!\!=\!\!&\!\! \sum_{j\in \N_i}b_{ij}(x_j - x_i),
\label{eq:2.9b}
\end{eqnarray}
where $u_i \in \R^N$ is an external input, $\xi_i \in \R^N$
is an additional variable which generates the integral of the consensus input 
$\sum_{j\in \N_i}b_{ij}(x_j - x_i)$, and 
$b_{ij}$ is the $(i,j)$-element of an adjacency matrix for $G$.
The graph Laplacian associated with the adjacency matrix with elements $b_{ij}$
is denoted by $L_{\I}$.

\begin{figure}
\centering
\includegraphics[width=7.5cm]{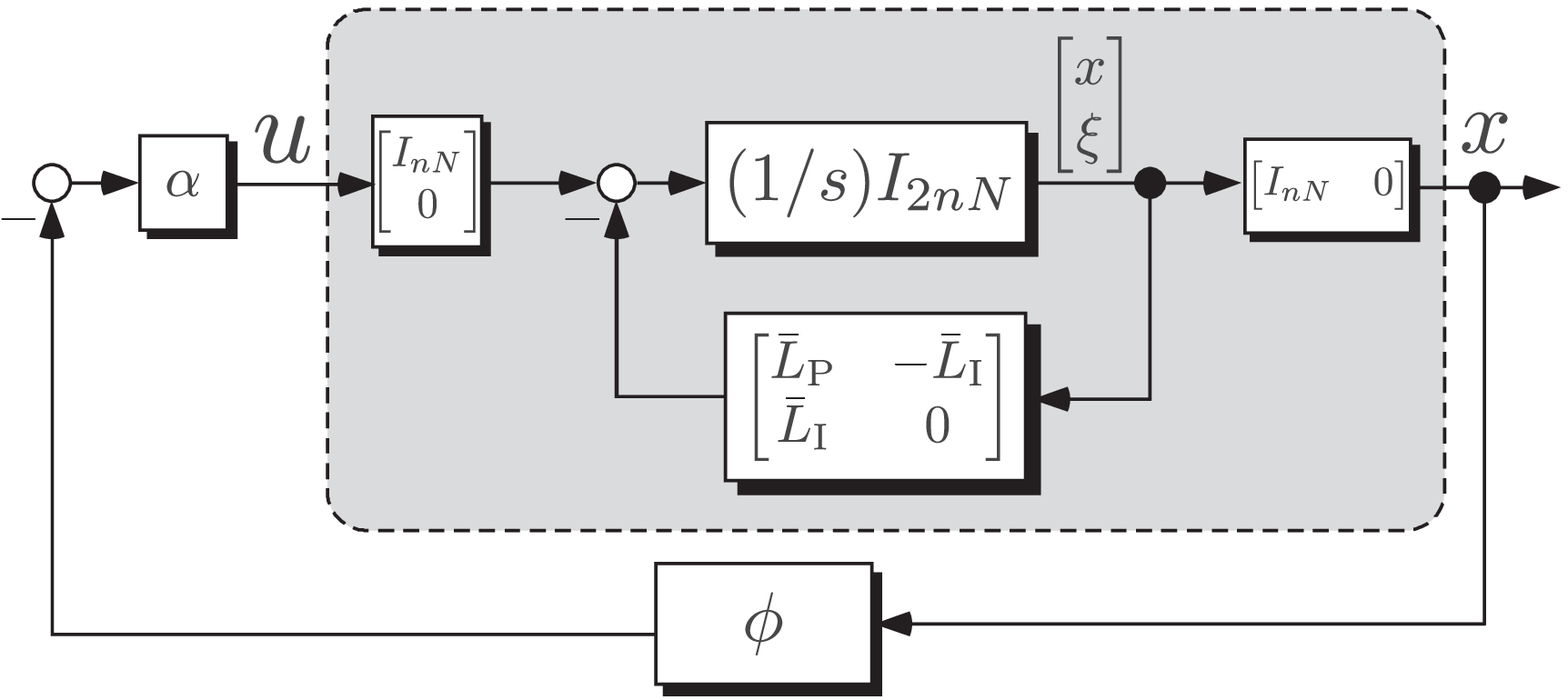}
\caption{Block diagram of the PI consensus-based distributed optimization algorithm for unconstrained optimization.
The system enclosed by the dashed line is passive from 
$\tilde u = u - u^*$ to
$\tilde x = x - x^*$ with $u^* := -\alpha \phi(x^*)$ as confirmed in (\ref{eq:2.21}).
The bottom block $\phi$ is incrementally passive (Lemma \ref{lem:0}), and hence passive from
$x - x^*$ to $\phi(x) - \phi(x^*)$.
}
\label{fig:2}
\end{figure}

Defining $\xi := [\xi_1^T\ \xi_2^T\ \cdots\ \xi_n^T]^T$ and 
$u := [u_1^T\ u_2^T\ \cdots\ u_n^T]^T$, the total system  is described as
\begin{eqnarray}
\begin{bmatrix}
\dot x\\
\dot \xi
\end{bmatrix}
\!\!&\!\!=\!\!&\!\! 
-
\begin{bmatrix}
\bar L_{\P}&-\bar L_{\I}\\
\bar L_{\I}& 0
\end{bmatrix}
\begin{bmatrix}
x\\
\xi
\end{bmatrix}
+ 
\begin{bmatrix}
I_{nN}\\
0
\end{bmatrix}
u,
\label{eq:2.10}
\end{eqnarray}
where $\bLI := \LI \otimes I_N$.
We can prove the following lemma
by making use of the skew symmetry of the non-diagonal blocks of
$\begin{bmatrix}
\bLP&-\bLI\\
\bLI& 0
\end{bmatrix}$.

\begin{lemma}
\label{lem:1}
The system (\ref{eq:2.10}) is passive from $u$ to $x$
with respect to the storage function
$S = \frac{1}{2}\|x\|^2 + \frac{1}{2}\|\xi\|^2$.
\end{lemma}
\begin{proof}
The time derivative of $S$ along the system trajectories
is given by
\begin{eqnarray}
\dot S \!\!&\!\!=\!\!&\!\! -\begin{bmatrix}
x\\
\xi
\end{bmatrix}^T
\begin{bmatrix}
\bLP&-\bLI\\
\bLI& 0
\end{bmatrix}
\begin{bmatrix}
x\\
\xi
\end{bmatrix}
+ 
\begin{bmatrix}
x\\
\xi
\end{bmatrix}^T
\begin{bmatrix}
I_{nN}\\
0
\end{bmatrix}
u
\nonumber\\
\!\!&\!\!=\!\!&\!\! -x^T \bLP x + x^Tu \leq x^T u.
\nonumber
\end{eqnarray}
This completes the proof.
\end{proof}

We close the loop of $u$
by the gradient-based feedback law 
$u = -\alpha \phi(x)$.
Then, the closed-loop system is formulated as
\begin{eqnarray}
\begin{bmatrix}
\dot x\\
\dot \xi
\end{bmatrix}
\!\!&\!\!=\!\!&\!\! 
-
\begin{bmatrix}
\bLP&-\bLI\\
\bLI& 0
\end{bmatrix}
\begin{bmatrix}
x\\
\xi
\end{bmatrix}
- \alpha
\begin{bmatrix}
I_{nN}\\
0
\end{bmatrix}
\phi(x),
\label{eq:2.14}
\end{eqnarray}
whose block diagram is illustrated in Fig. \ref{fig:2}.



\subsection{Equilibrium/Convergence Analysis: Unconstrained Case}
\label{sec:2.3}

Regarding the equilibrium of (\ref{eq:2.14}), we have the following lemma.
\begin{lemma}
\label{lem:2}
Under Assumptions \ref{ass:1} and \ref{ass:2}, there exists $\xi^*$ such that 
$[(x^*)^T\ (\xi^*)^T]^T$ is an equilibrium of (\ref{eq:2.14}).
\end{lemma}
\begin{proof}
See \cite{WE_AAC10}.
\end{proof}

Lemma \ref{lem:2} means that, if we define $u^* := -\alpha \phi(x^*)$,
the following equation holds.
\begin{eqnarray}
(0=)
\begin{bmatrix}
\dot x^*\\
\dot \xi^*
\end{bmatrix}
\!\!&\!\!=\!\!&\!\! 
-
\begin{bmatrix}
\bLP&-\bLI\\
\bLI& 0
\end{bmatrix}
\begin{bmatrix}
x^*\\
\xi^*
\end{bmatrix}
+ 
\begin{bmatrix}
I_{nN}\\
0
\end{bmatrix}
u^*
\label{eq:2.16}
\end{eqnarray}
Define $\tilde x := x-x^*$, $\tilde \xi := \xi - \xi^*$ and $\tilde u := u - u^*$.
Then, subtracting (\ref{eq:2.16}) from (\ref{eq:2.10}) yields
\begin{eqnarray}
\begin{bmatrix}
\dot{\tilde{x}}\\
\dot{\tilde{\xi}}
\end{bmatrix}
\!\!&\!\!=\!\!&\!\! 
-
\begin{bmatrix}
\bLP&-\bLI\\
\bLI& 0
\end{bmatrix}
\begin{bmatrix}
\tilde x\\
\tilde \xi
\end{bmatrix}
+ 
\begin{bmatrix}
I_{nN}\\
0
\end{bmatrix}
\tilde u.
\label{eq:2.17}
\end{eqnarray}
Since the system matrices of (\ref{eq:2.17})
are the same as those of (\ref{eq:2.10}),
passivity of the system (\ref{eq:2.17}) 
from $\tilde u$ to $\tilde x$ can be immediately proved by
modifying the storage function as
\begin{eqnarray}
\tilde S := \frac{1}{2}\|\tilde x\|^2 + \frac{1}{2}\|\tilde \xi\|^2.
\label{eq:2.18}
\end{eqnarray}
More precisely, the following inequality holds.
\begin{eqnarray}
\dot{\tilde{S}} \!\!&\!\!=\!\!&\!\! -\tilde x^T \bLP \tilde x + \tilde x^T\tilde u \leq \tilde x^T\tilde u.
\label{eq:2.21}
\end{eqnarray}
Remark that, differently from (\ref{eq:0.10}), the input and output
are defined by the increments $\tilde x := x-x^*$ and $\tilde u = u - u^*$.


We are now ready to use the passivity interpretation 
of the system (\ref{eq:2.10}) and gradient of convex functions to
prove the following convergence result.

\begin{theorem}
\label{thm:1}
Consider the system (\ref{eq:2.14}).
If Assumptions \ref{ass:1} and \ref{ass:2} hold, 
then $x_i$ asymptotically converges to the set of optimal solutions to (\ref{eq:2.1})
for all $i = 1,2,\dots, n$.
\end{theorem}
\begin{proof}
From $u = -\alpha \phi(x)$ and $u^* = -\alpha \phi(x^*)$,
(\ref{eq:2.21}) is rewritten as 
\begin{eqnarray}
\dot{\tilde{S}} 
\!\!&\!\!=\!\!&\!\! -(x-x^*)^T \bLP (x-x^*) 
- \alpha(x-x^*)^T(\phi(x)-\phi(x^*)).
\nonumber
\end{eqnarray}
Since $\bLP x^* = 0$,  
this is further rewritten as 
\begin{eqnarray}
\dot{\tilde{S}} 
\!\!&\!\!=\!\!&\!\! -x^T \bLP x - \alpha \sum_{i=1}^n (x_i-z^*)^T(\phi_i(x_i)-\phi_i(z^*))
\nonumber
\end{eqnarray}
Using Lemma \ref{lem:0} and $\bLP\geq 0$, 
we can prove $\dot{\tilde{S}} \leq 0$.

Since the function $\tilde S$ is radially unbounded and positive definite, 
any level set of the function is positively invariant.
Hence LaSalle's principle is applicable.
Consider the state trajectories such that $\dot{\tilde{S}} \equiv 0$, i.e.,
both of  $x^T \bLP x = 0$ and
\begin{eqnarray}
\sum_{i=1}^n (x_i-z^*)^T(\phi_i(x_i)-\phi_i(z^*))=0
\label{eq:2.25b}
\end{eqnarray}
identically hold. The former equation  means
consensus of the variable $x_i$, namely there exists $c(\cdot)$
such that $x_i(t) = c(t)\ {\forall i}, t$.
In this case, (\ref{eq:2.25b}) is rewritten as
\begin{eqnarray}
\!\!\!\!\!\!\!\!\!\!
0 
\!\!&\!\!=\!\!&\!\! (c(t)-z^*)^T \sum_{i=1}^n (\phi_i(c(t))-\phi_i(z^*))
\nonumber\\
\!\!\!\!\!\!\!\!\!\!
\!\!&\!\!=\!\!&\!\! (c(t)-z^*)^T \sum_{i=1}^n\phi_i(c(t))
=(c(t)-z^*)^T \nabla f(c(t)),
\label{eq:2.26}
\end{eqnarray}
where the third equation holds from the optimality condition (\ref{eq:2.2}).
Using (\ref{eq:0.4}), (\ref{eq:2.26}) is further rewritten as
\begin{eqnarray}
0 \!\!&\!\!=\!\!&\!\! (\nabla f(c(t)))^T(c(t)-z^*) \geq f(c(t)) - f(z^*) \geq 0,
\nonumber
\end{eqnarray}
which implies that $f(c(t)) = f(z^*)$ holds for all $t$, namely the trajectories of $c$ 
must be contained in the set of optimal solutions to (\ref{eq:2.1}).
Thus, LaSalle's invariance principle proves the theorem.
\end{proof}


\begin{remark}
The above algorithm together with the 
convergence result compatible with ours was already presented in 
\cite{WE_AAC10}.
The contribution of this section is not to prove convergence itself but to
provide a passivity-based perspective 
that (\ref{eq:2.14}) is regarded as feedback connection of 
two passive systems with incremental inputs and outputs.
It will be shown in the subsequent sections that this perspective provides 
fruitful design concepts.
\end{remark}

\section{Unconstrained Distributed Optimization with Communication Delay}
\label{sec:3}

In this section, we suppose that the inter-agent communication suffers from delays
which are assumed to be constant but heterogeneous.
The delay
from agent $i$ to $j$ is denoted 
by $T_{ij}$ for any pair $(i,j)\in \E$.

We start with the following restrictive assumption
in order to enhance readability, and then it will be relaxed.
\begin{assumption}
\label{ass:5}
The functions $f_1, f_2, \dots, f_n$ are strictly convex, continuously differentiable,
and their gradients are locally Lipschitz.
\end{assumption}
Under the assumption together with the existence of the optimal solution,
the solution $z^*$ is uniquely determined \cite{B_BK}.

\begin{figure}
\centering
\includegraphics[width=8.4cm]{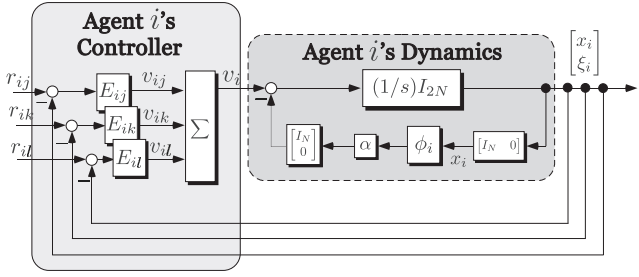}
\caption{Block diagram of agent $i$'s dynamics for unconstrained optimization in the presence of communication delays.
The system enclosed by the dashed line is 
passive from $\bar v_i$
to $[{\bar{x}}_i^T\ {\bar{\xi}}_i^T]^T$ (Lemma \ref{lem:18}), where $\bar v_i$, $\bar x_i$ and $\bar \xi_i$
are defined in (\ref{eq:7.2a}) and (\ref{eq:8.12a}).
}
\label{fig:3}
\end{figure}

\subsection{Individual Dynamics as A Feedback System with Passive Dynamics and A Collection of Passive Controllers}

Consider the PI consensus algorithm (\ref{eq:2.9a}) and (\ref{eq:2.9b}) with $u_i = -\alpha \phi_i(x_i)$,
where $x_j$ and $\xi_j$ are sent by agent $j \in \N_i$ through communication.
Since these information are not directly available in the presence of delays,
we replace these variables by new notations $r^x_{ij}$ and $r^{\xi}_{ij}$ as
\begin{eqnarray}
\!\!\!\!\!\!\!\!
\dot x_i \!\!&\!\!=\!\!&\!\! \sum_{j\in \N_i}a_{ij}(r^x_{ij} - x_i) - \sum_{j\in \N_i}b_{ij}(r^{\xi}_{ij} - \xi_i) -  \alpha \phi_i(x_i),
\label{eq:8.1a}\\
\!\!\!\!\!\!\!\!
\dot \xi_i \!\!&\!\!=\!\!&\!\! \sum_{j\in \N_i}b_{ij}(r^x_{ij} - x_i),
\label{eq:8.1b}
\end{eqnarray}
where $r^x_{ij}$ stands for the signal received by agent $i$ that
is sent by agent $j$ regarding the value of $x_j$.
The block diagram of the system is illustrated in Fig. \ref{fig:3}, which is regarded as a feedback system 
with the agent dynamics
\begin{eqnarray}
\begin{bmatrix}
\dot{{x}}_i\\
\dot{{\xi}}_i
\end{bmatrix}
= v_i -  \alpha 
\begin{bmatrix}
\phi_i(x_i)\\
0
\end{bmatrix}
\label{eq:8.55}
\end{eqnarray}
and the controller
\begin{eqnarray}
v_i  \!\!&\!\!=\!\!&\!\! \sum_{j\in \N_i}v_{ij},\ 
v_{ij} := 
\begin{bmatrix}
v_{ij}^x\\
v_{ij}^{\xi}
\end{bmatrix} =
E_{ij}
\begin{bmatrix}
r^x_{ij} - x_i\\ 
r^{\xi}_{ij} - \xi_i
\end{bmatrix},
\label{eq:8.11}\\
E_{ij}  \!\!&\!\!=\!\!&\!\!
\begin{bmatrix}
a_{ij}I_N&-b_{ij}I_N\\
b_{ij}I_N&0
\end{bmatrix}\ \ j \in \N_i.
\nonumber
\end{eqnarray}
It is easy to confirm that $E_{ij}$ is a passive map.

Let us cut the loop of Fig. \ref{fig:3} and
focus on the open-loop system (\ref{eq:8.55})  with input $v_i$
enclosed by the dashed line in Fig. \ref{fig:3}.
Then, we have the following lemma.

\begin{lemma}
\label{lem:18}
Suppose that Assumption \ref{ass:5} holds.
Then, the system (\ref{eq:8.55}) is passive from 
$\bar v_{i}$ to $[{\bar{x}}_i^T\ {\bar{\xi}}_i^T]^T$
with respect to the storage function
$\bar S_i := \frac{1}{2}\|\bar x_i\|^2 + \frac{1}{2}\|\bar \xi_i\|^2$,
where
\begin{eqnarray}
\bar x_i \!\!&\!\!:=\!\!&\!\! x_i-z^*,\ \bar \xi_i := \xi_i - 2\xi_i^*, 
\label{eq:7.2a}\\
\bar v_{i} \!\!&\!\!:=\!\!&\!\! v_{i} - \sum_{j\in \N_i}
v^*_{ij},\
v^*_{ij} :=
b_{ij}
\begin{bmatrix}
\xi_i^* - \xi_j^*\\
0
\end{bmatrix}.
\label{eq:8.12a}
\end{eqnarray}
The notation $\xi_i^* \in \R^{N}$ is a vector such that 
the stack vector $[(\xi_1^*)^T\ \cdots\ (\xi_n^*)^T]^T$ 
is equal to $\xi^*$ in Lemma \ref{lem:2}.
\end{lemma}

\begin{proof}
From (\ref{eq:2.16}), it follows that
\begin{eqnarray}
\!\!\!\!\!\!\!\!\!\!
\dot x_i^* \!\!&\!\!=\!\!&\!\!
\sum_{j\in \N_i}a_{ij}(z^* - z^*) - \sum_{j\in \N_i}b_{ij}(\xi_j^* - \xi_i^*) 
-  \alpha \phi_i(z^*)
\nonumber\\
\!\!\!\!\!\!\!\!\!\!
\!\!&\!\!=\!\!&\!\! 
\sum_{j\in \N_i}b_{ij}(\xi_i^* - \xi_j^*) 
-  \alpha \phi_i(z^*),
\label{eq:8.2a}\\
\!\!\!\!\!\!\!\!\!\!
\dot \xi_i^* \!\!&\!\!=\!\!&\!\! \sum_{j\in \N_i}b_{ij}(z^* - z^*) = 0.
\label{eq:8.2b}
\end{eqnarray}
Subtracting (\ref{eq:8.2a}) and (\ref{eq:8.2b}) from (\ref{eq:8.55}) yields
\begin{eqnarray}
\begin{bmatrix}
\dot{\bar{x}}_i\\
\dot{\bar{\xi}}_i
\end{bmatrix}
= \bar v_i -  \alpha 
\begin{bmatrix}
\phi_i(x_i)-\phi_i(z^*)\\
0
\end{bmatrix}.
\label{eq:8.5}
\end{eqnarray}
because of the definition of $\bar v_i$ in (\ref{eq:8.12a}).
The time derivative of $\bar S_i$ along the trajectories of (\ref{eq:8.5}) is then given by
\begin{eqnarray}
\dot{\bar{S}}_i = \begin{bmatrix}
{\bar{x}}_i\\
{\bar{\xi}}_i
\end{bmatrix}^T \bar v_i - \alpha(x_i - z^*)^T(\phi_i(x_i)-\phi_i(z^*)).
\label{eq:8.6}
\end{eqnarray}
Lemma \ref{lem:0} and Assumption \ref{ass:5} prove the lemma.
\end{proof}

Let us now close the loop between (\ref{eq:8.55}) and the controller (\ref{eq:8.11}).
Define 
\begin{eqnarray}
\bar r_{ij} :=
\begin{bmatrix}
\bar r^x_{ij}\\
\bar r^{\xi}_{ij}
\end{bmatrix}
=
 \begin{bmatrix}
r^x_{ij}\\
r^{\xi}_{ij} 
\end{bmatrix} - r^*_{ij}, \
r^*_{ij} :=
\begin{bmatrix}
z^*\\
\xi_i^* + \xi_j^*
\end{bmatrix}.
\label{eq:8.12b}
\end{eqnarray}
and $\bar v_{ij} := v_{ij} - v_{ij}^*$.
Then, from (\ref{eq:8.11}), (\ref{eq:7.2a}) and (\ref{eq:8.12a}), we have
\begin{eqnarray}
\bar v_{ij} \!\!&\!\!=\!\!&\!\! 
E_{ij}\begin{bmatrix}
 r^x_{ij} - x_i\\
r^{\xi}_{ij} - \xi_i
\end{bmatrix} -
b_{ij}
\begin{bmatrix}
\xi_i^* - \xi_j^*\\
0
\end{bmatrix}
\nonumber\\
\!\!&\!\!=\!\!&\!\! 
E_{ij}\begin{bmatrix}
 r^x_{ij} - x_i\\
r^{\xi}_{ij} - \xi_i + (\xi_i^* - \xi^*_j)
\end{bmatrix}
\nonumber\\
\!\!&\!\!=\!\!&\!\!
E_{ij}\begin{bmatrix}
(r^x_{ij}-z^*) - (x_i-z^*)\\
(r^{\xi}_{ij} - (\xi_i^* + \xi^*_j)) - (\xi_i - 2\xi_i^*) 
\end{bmatrix}
=
E_{ij}\begin{bmatrix}
\bar r^x_{ij} - \bar x_i\\
\bar r^{\xi}_{ij} - \bar \xi_i
\end{bmatrix}.
\nonumber
\end{eqnarray}
Substituting this together with $\bar v_i = \sum_{j\in \N_i}\bar v_{ij}$ into (\ref{eq:8.6}) 
proves the 
following passivity-like property of the closed-loop system.
\begin{eqnarray}
\!\!\!\!\!\!
\dot{\bar{S}}_i  \!\!&\!\!=\!\!&\!\!   \sum_{j\in \N_i}
\begin{bmatrix}
{\bar{x}}_i\\
{\bar{\xi}}_i
\end{bmatrix}^T
\begin{bmatrix}
a_{ij}(\bar r^x_{ij} - \bar x_i) + b_{ij}(\bar \xi_i - \bar r^{\xi}_{ij})\\
b_{ij}(\bar r^x_{ij} - \bar x_i)
\end{bmatrix}
\nonumber\\
\!\!\!\!\!\!&& \hspace{1cm}- \alpha(x_i - z^*)^T(\phi_i(x_i)-\phi_i(z^*))
\nonumber\\
\!\!\!\!\!\!
\!\!&\!\!=\!\!&\!\!   \sum_{j\in \N_i}\{a_{ij}
({\bar{x}}_i^T \bar r^x_{ij} - \|{\bar{x}}_i\|^2) + 
b_{ij}(-{\bar{x}}_i^T\bar r^{\xi}_{ij} + {\bar{\xi}}_i^T\bar r^x_{ij})\}
 \nonumber\\
\!\!\!\!\!\!&& \hspace{1cm}- \alpha(x_i - z^*)^T(\phi_i(x_i)-\phi_i(z^*))
\nonumber\\
 \!\!\!\!\!\!
\!\!&\!\!=\!\!&\!\! \sum_{j\in \N_i}
\bar r_{ij}^T \bar v_{ij} - \sum_{j\in \N_i}a_{ij}\|{\bar{x}}_i - \bar r^x_{ij}\|^2
\nonumber\\
\!\!\!\!\!\!&& - \alpha(x_i - z^*)^T(\phi_i(x_i)-\phi_i(z^*))
\leq \sum_{j\in \N_i}
\bar r_{ij}^T \bar v_{ij}. 
\label{eq:8.9}
\end{eqnarray}
Remark that, from (\ref{eq:8.12a}) and (\ref{eq:8.12b}), the following equations hold,
which plays an important role in deriving the subsequent theoretical results in the section.
\begin{eqnarray}
v^*_{ji} = -v^*_{ij},\  r^*_{ji} = r^*_{ij}.
\label{eq:8.20}
\end{eqnarray}

\subsection{Scattering Transformation}

\begin{figure}
\centering
\includegraphics[width=8.4cm]{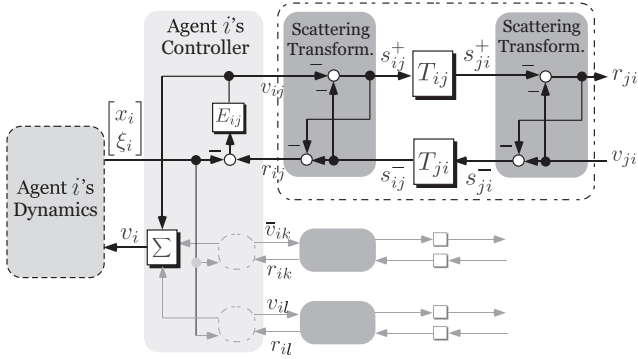}
\caption{Block diagram of agent $i$'s dynamics including the scattering transformation and 
communication delays.
The system enclosed by the dash-dotted line is passive from 
$-[\bar v_{ij}^T\ \bar v_{ji}]^T$ to $[\bar r_{ij}^T\ \bar r_{ji}^T]^T$ (Lemma \ref{lem:com}), where $\bar v_{ij}$ and $\bar r_{ij}$
are defined in (\ref{eq:8.12a}) and (\ref{eq:8.12b}).}
\label{fig:4}
\end{figure}

Our problem is now regarded as a synchronization problem of the variable
$x_i\ (i=1,2,\dots, n)$ to the optimal solution $z^*$.
In the field of output synchronization, communication delays have been investigated by many publications
(See \cite{HCFS_BK} and references therein). 
Among them, this paper focuses on the approach of \cite{CS_CDC06}, where the authors
present a passivity-based control architecture with a form similar to that in the previous subsection.
To be concrete, a closed-loop system is firstly formed with a passive dynamics,
(\ref{eq:8.55}) in the present case, and a collection of passive controllers, (\ref{eq:8.11}) in the present case.
Then, based on the passivity-like property (\ref{eq:8.9}),
the controller output is exchanged with neighboring agents through scattering transformation, introduced below,
and then synchronization is proved. 
This inspires us to exchange the controller outputs $\bar v_{ij}$ instead of $x_i$ and $\xi_i$.
However, in the case of the present problem, $\bar v_{ij}$ ensuring
the passivity-like property (\ref{eq:8.9})
includes $\xi_i^*$ and $\xi_j^*$ in the definitions of 
$\bar \xi_i$ and $\bar r^{\xi}_{ij}$
which are not available for agent $i$.
Thus, we instead let agent $i$ send 
$v_{ij}$
by eliminating such unavailable terms. 
%

Following \cite{CS_CDC06},
the message $v_{ij}$ is sent to neighbors $j \in \N_i$ through
the scattering transformation,
which is well-known to passify the communication block including the delays \cite{HCFS_BK}.
In the present case, the transformation is defined as
\begin{eqnarray}
\!\!\!\!\!\!\!\!\!\!
s_{ij}^{\ra} \!\!&\!\!=\!\!&\!\! \frac{1}{\sqrt{2\eta}}(-v_{ij} + \eta r_{ij}),\ s_{ij}^{\la} = \frac{1}{\sqrt{2\eta }}(-v_{ij} - \eta r_{ij}),
\label{eq:8.13a}\\
\!\!\!\!\!\!\!\!\!\!
s_{ji}^{\la} \!\!&\!\!=\!\!&\!\! \frac{1}{\sqrt{2\eta }}(v_{ji} + \eta r_{ji}),\
s_{ji}^{\ra} = \frac{1}{\sqrt{2\eta }}(v_{ji} - \eta r_{ji}),
\label{eq:8.13b}
\end{eqnarray}
where $r_{ij} := [(r_{ij}^x)^T\ (r_{ij}^{\xi})^T]^T$ and $\eta > 0$.
Then, as illustrated in Fig. \ref{fig:4}, agent $i$ sends $s_{ij}^{\ra}$ instead of $v_{ij}$
and it is received by agent $j$ as $s_{ji}^{\la}$ after the delay $T_{ij}$.
On the other hand, agent $j$ sends $s_{ji}^{\ra}$ 
and it is received by agent $i$ as $s_{ij}^{\la}$ after the delay $T_{ji}$.
Namely, these signals satisfy the following relations.
\begin{eqnarray}
s_{ji}^{\la}(t) =
s_{ij}^{\ra}(t-T_{ij}),\
s_{ij}^{\la}(t) =
s_{ji}^{\ra} (t-T_{ji}).
\label{eq:8.21}
\end{eqnarray}
Once agent $i$ receives $s_{ij}^{\la}$, it computes
$r_{ij}$ from the second equation of (\ref{eq:8.13a})
and adds the resulting $r_{ij}$ to the controller
(\ref{eq:8.11}).
Note that the operation of agent $i$ differs from that of agent $j$, and hence
any pair $(i,j) \in \E$ has to share which operation to be executed, which can be done
by using IDs $i = 1,2,\dots, n$.
An algorithmic description of agent $i$'s operation is shown in Algorithm \ref{alg:1}.

\begin{algorithm}[t]
\caption{Algorithmic Description of in Agent $i$'s Operation in Scattering Transformation}
\label{alg:1}
\begin{algorithmic}
\FOR{$j \in \N_i$}
\IF{$j < i$}
\STATE Set $s_{ij}^{\ra} \la \frac{1}{\sqrt{2\eta}}(-v_{ij} + \eta r_{ij})$ 
and send it to $j$.
\ELSIF{$j > i$}
\STATE Set $s_{ij}^{\ra} \la \frac{1}{\sqrt{2\eta }}(v_{ij} - \eta r_{ij})$
and send it to $j$.
\ENDIF
\ENDFOR
\STATE  Receive Messages $s_{ij}^{\la}$ from all neighbors $j\in \N_i$
\FOR{$j \in \N_i$}
\IF{$j < i$}
\STATE  $r_{ij} \la -(I_{2N}-(1/\eta)E_{ij})^{-1}(\frac{\sqrt{2}}{\sqrt{\eta}} s_{ij}^{\la} + \frac{1}{\eta}E_{ij}[x_i^T\ \xi_i^T]^T)$
\ELSIF{$j > i$}
\STATE  $r_{ij} \la (I_{2N} + (1/\eta)E_{ij})^{-1}(\frac{\sqrt{2}}{\sqrt{\eta}} s_{ij}^{\la} + \frac{1}{\eta}E_{ij}[x_i^T\ \xi_i^T]^T)$
\ENDIF
\ENDFOR
\STATE Set $v_i \la \sum_{j\in \N_i}E_{ij}(r_{ij} - [x_i^T\ \xi_i^T]^T)$ and input $v_i$ to its own dynamics
\end{algorithmic}
\end{algorithm}

The system with the scattering 
transformation (\ref{eq:8.13a}) and (\ref{eq:8.13b}) and the delay blocks (\ref{eq:8.21})
is known to be passive from 
$-[v_{ij}^T\ v_{ji}]^T$ to $[r_{ij}^T\ r_{ji}^T]^T$ \cite{HCFS_BK}.
The following lemma proves that the system is also passive
from the input to output with the biases $v_{ij}^*$, $v_{ji}^*$, $r_{ij}^*$ and $r_{ji}^*$.
For simplicity, we suppose that $s_{ij}^{\ra}(t) = s_{ji}^{\ra}(t) = 0\ {\forall t < 0}$ throughout this paper.

\begin{lemma}
\label{lem:com}
The system consisting of the scattering 
transformation (\ref{eq:8.13a}) and (\ref{eq:8.13b}) and the delay blocks (\ref{eq:8.21})
is passive from $-[\bar v_{ij}^T\ \bar v_{ji}]^T$ to $[\bar r_{ij}^T\ \bar r_{ji}^T]^T$.
\end{lemma}
\begin{proof}
Define
\begin{eqnarray}
V_{ij} \!\!&\!\!:=\!\!&\!\ \frac{1}{2}\int^t_0 \Big(\|s_{ij}^{\ra} + 
\frac{1}{\sqrt{2\eta }}(v^*_{ij} - \eta r^*_{ij})\|^2 
\nonumber\\
\!\!&\!\!\!\!&\!\
- \|s_{ji}^{\la} + \frac{1}{\sqrt{2\eta }}(v^*_{ij} - \eta r^*_{ij})\|^2
\nonumber\\
\!\!&\!\!\!\!&\!\
+\|s_{ji}^{\ra} + \frac{1}{\sqrt{2\eta }}(v^*_{ij} + \eta r^*_{ij})\|^2 
\nonumber\\
\!\!&\!\!\!\!&\!\
- \|s_{ij}^{\la} + \frac{1}{\sqrt{2\eta }}(v^*_{ij} + \eta r^*_{ij})\|^2\Big)d\tau
\nonumber\\
\!\!&\!\!\!\!&\!\
+ \frac{T_{ij}}{4\eta }\|v^*_{ij} - \eta r^*_{ij}\|^2 + \frac{T_{ji}}{4\eta }\|v^*_{ij} + \eta r^*_{ij}\|^2.
\label{eq:8.14}
\end{eqnarray}
Since $v^*_{ij}$ and $r^*_{ij}$ are both constant, 
(\ref{eq:8.14}) is rewritten as
\begin{eqnarray}
\hspace{-.5cm}V_{ij} \!\!&\!\!=\!\!&\!\ \frac{1}{2}\int^t_{t-T_{ij}} \|s_{ij}^{\ra} + 
\frac{1}{\sqrt{2\eta }}(v^*_{ij} - \eta r^*_{ij})\|^2 d\tau
\nonumber\\
\hspace{-.5cm}\!\!&\!\!\!\!&\!\
+ \frac{1}{2}\int^t_{t-T_{ji}}
\|s_{ji}^{\ra} + \frac{1}{\sqrt{2\eta }}(v^*_{ij} + \eta r^*_{ij})\|^2 d\tau\geq 0.
\nonumber
\end{eqnarray}

Using (\ref{eq:8.13a}) and (\ref{eq:8.13b}), the time derivative of $V_{ij}$ is given by
\begin{eqnarray}
\dot V_{ij} \!\!&\!\!=\!\!&\!\ \frac{1}{4\eta }\Big(
\left\|(-v_{ij} + v^*_{ij}) + \eta  (r_{ij}-r^*_{ij})\right\|^2 
\nonumber\\
\!\!&\!\!\!\!&\!\
- \left\|(v_{ji} + v^*_{ij}) + \eta (r_{ji} - r^*_{ij})\right\|^2
\nonumber\\
\!\!&\!\!\!\!&\!\
+\left\|(v_{ji} + v^*_{ij}) - \eta (r_{ji} - r^*_{ij})\right\|^2 
\nonumber\\
\!\!&\!\!\!\!&\!\
- \left\|(-v_{ij} + v^*_{ij}) - \eta  (r_{ij} - r^*_{ij})\right\|^2\Big).
\nonumber
\end{eqnarray}
From (\ref{eq:8.12a}), (\ref{eq:8.12b})  and (\ref{eq:8.20}), we also have
\begin{eqnarray}
\dot V_{ij} \!\!&\!\!=\!\!&\!\ \frac{1}{4\eta }\Big(
\left\|-\bar v_{ij} + \eta  \bar r_{ij}\right\|^2 
- \left\|\bar v_{ji} + \eta \bar r_{ji} \right\|^2
\nonumber\\
\!\!&\!\!\!\!&\!\
+\left\|\bar v_{ji} - \eta \bar r_{ji} \right\|^2 
- \left\|-\bar v_{ij} - \eta  \bar r_{ij}\right\|^2\Big)
\nonumber\\
\!\!&\!\!=\!\!&\!\ - \bar v_{ij}^T \bar r_{ij} - \bar v_{ji}^T \bar r_{ji}.
\label{eq:8.16}
\end{eqnarray}
This completes the proof.
\end{proof}

\subsection{Convergence Analysis}

We are now ready to prove the convergence result.
\begin{lemma}
\label{thm:8}
Consider the system (\ref{eq:8.1a}) and (\ref{eq:8.1b}) for all $i$, and
the scattering transformation (\ref{eq:8.13a}) and (\ref{eq:8.13b}) and the delays (\ref{eq:8.21}) for all $j\in \N_i$ and all $i$.
If Assumptions \ref{ass:2} and \ref{ass:5} hold, 
then $x_i$ asymptotically converges to the optimal solution $z^*$
to (\ref{eq:2.1}) for all $i=1,2,\dots,n$.
\end{lemma}
\begin{proof}
Define 
\[
V := \sum_{i=1}^n \bar S_i + \sum_{(i,j)\in \E}V_{ij}.
\]
Then, combining (\ref{eq:8.9}) and (\ref{eq:8.16}), we obtain
\begin{eqnarray}
\dot V 
\!\!&\!\!=\!\!&\!\! - \sum_{i=1}^n\sum_{j\in \N_i}a_{ij}\|{\bar{x}}_i - \bar r^x_{ij}\|^2 
\nonumber\\
\!\!&\!\!\!\!&\!\!
 - \alpha\sum_{i=1}^n (x_i - z^*)^T(\phi_i(x_i)-\phi_i(z^*)) \leq 0.
\label{eq:8.17}
\end{eqnarray}
Now, define $X := [x^T\ \xi^T]^T$ and $X_t$ such that $X_t(\theta) = X(t+\theta)$ for $\theta \in [-\max_{i,j}T_{ij}, 0]$.
Then, the extension of the LaSalle's principle for time delay systems \cite{hale} is
applicable and any solution $X_t$ to the system converges to the largest invariant set in
the set of trajectories satisfying $\dot V \equiv 0$.
Under Assumption \ref{ass:5}, the gradient $\phi_i$ satisfies 
$(x_i - z^*)^T(\phi_i(x_i)-\phi_i(z^*)) > 0$ whenever $x_i \neq z^*$.
Thus, $\dot V = 0$ means $x_i \equiv z^*$ for all $i$,
and hence we can conclude $x_i \to z^*\ {\forall i} = 1,2,\dots, n$.
\end{proof}

The above proof, in particular (\ref{eq:8.17}), means that the communication delay blocks 
are successfully integrated with the distributed optimization algorithm 
as interconnections of passive systems.
However, Assumption \ref{ass:5} requires that $f_i$ is strictly convex on
all elements of $z$,
which is fairly strong and may limit applications.
The assumption can be relaxed as below.
Suppose now that $f_i$ depends only on $(z_l)_{l \in {\mathcal Z}_i}$ for a subset
$ {\mathcal Z}_i \subseteq \{1,2,\dots, N\}$, where $z_l$ is the $l$-th element of $z$.
Then, we assume the following.
\begin{assumption}
\label{ass:11}
The functions $f_1, f_2, \dots, f_n$ are continuously differentiable and their gradients are locally Lipschitz. 
Every $f_i\ (i = 1,2,\dots,n)$ is strictly convex in $(z_l)_{l \in {\mathcal Z}_i}$, where $ {\mathcal Z}_i \subseteq \{1,2,\dots, N\}$. 
Also, $\cup_{i=1}^N {\mathcal Z}_i = \{1,2,\dots, N\}$ holds.
\end{assumption}
Remark that  $ {\mathcal Z}_i$ can be empty for some $i$.

Under Assumption \ref{ass:11} instead of \ref{ass:5}, we have the following theorem.
\begin{theorem}
\label{thm:10}
Consider the same system as Lemma \ref{thm:8}.
If Assumptions \ref{ass:2} and \ref{ass:11} hold, 
then $x_i$ asymptotically converges to the optimal solution $z^*$ to
(\ref{eq:2.1}) for all $i = 1,2,\dots, n$.
\end{theorem}
\begin{proof}
Suppose that $l \in {\mathcal Z}_i$.
It is then sufficient to prove that the $l$-th element of $x_{j}$ converges to $z^*_l$ for all $j$, 
where $z^*_l \in \R$ is the $l$-th element of $z^*$. 
Similarly to Lemma \ref{thm:8}, LaSalle's principle for time delay systems \cite{hale} is applicable and hence
we consider the set of solutions satisfying $\dot V \equiv 0$.
In the set, $\|{\bar{x}}_i - \bar r^x_{ij}\| = 0\ {\forall j}\in \N_i$ holds, which means $\dot \xi = 0$.
Thus, LaSalle's principle implies, under Assumption \ref{ass:11} and $l \in {\mathcal Z}_i$, that
\begin{eqnarray}
\mbox{(i) }\ \ \lim_{t\to \infty}(x_i - r^x_{ij}) = 0\ {\forall j}\in \N_i,\ {\forall i}=1,2,\dots,n
\label{eq:8.32a}
\end{eqnarray}
(ii) the $l$-th element of $x_i$ converges to $z^*_l$,
and (iii) $\xi_i$ has a limit $\lim_{t\to \infty} \xi_{i} $.
From (\ref{eq:8.11}) and (\ref{eq:8.13a})--(\ref{eq:8.21}),
we obtain
\begin{eqnarray}
r_{ij}^x  \!\!&\!\!=\!\!&\!\! r_{ji}^x(t-T_{ji}) + (1/\eta)\{
- a_{ij}d_{ij} + b_{ij}(r_{ij}^{\xi} - \xi_i) 
\nonumber\\
\!\!&\!\!\!\!&\!\!  \hspace{1.5cm}+ b_{ij}(r_{ji}^{\xi}(t-T_{ji}) - \xi_j(t-T_{ji}))\},
\label{eq:8.33a}\\
r_{ji}^x  \!\!&\!\!=\!\!&\!\! r_{ij}^x(t-T_{ij}) + (1/\eta)\{
- a_{ij}d_{ji} + b_{ij}(r_{ji}^{\xi} - \xi_j) 
\nonumber\\
\!\!&\!\!\!\!&\!\!  \hspace{1.5cm} + b_{ij}(r_{ij}^{\xi}(t-T_{ij}) - \xi_i(t-T_{ij}))\},
\label{eq:8.33b}\\
r^{\xi}_{ij} \!\!&\!\!=\!\!&\!\! r^{\xi}_{ji}(t-T_{ji})- (b_{ij}/\eta)d_{ij},
\label{eq:8.33c}\\
r^{\xi}_{ji} \!\!&\!\!=\!\!&\!\! r^{\xi}_{ij}(t-T_{ij})- (b_{ij}/\eta)d_{ji},
\label{eq:8.33d}
\end{eqnarray}
with 
\begin{eqnarray}
d_{ij}(t) \!\!&\!\!:=\!\!&\!\! (r_{ij}^x - x_i)  + (r_{ji}^x(t-T_{ji}) - x_j(t-T_{ji})),
\nonumber\\
d_{ji}(t) \!\!&\!\!:=\!\!&\!\! (r_{ji}^x - x_j) + (r_{ij}^x(t-T_{ij}) - x_i(t-T_{ij})).
\nonumber
\end{eqnarray}
Summing (\ref{eq:8.33d}) at time $t-T_{ji}$ and (\ref{eq:8.33c}) yields
\begin{eqnarray}
r^{\xi}_{ij} - r^{\xi}_{ij}(t-\bar T_{ij}) = - (b_{ij}/\eta)(d_{ij} + d_{ji}(t-T_{ji})),
\label{eq:8.37}
\end{eqnarray}
where $\bar T_{ij} := T_{ij} + T_{ji}$.
From (\ref{eq:8.32a}), 
\begin{eqnarray}
\lim_{t\to \infty} d_{ij} = 0 \mbox{ and } \lim_{t\to \infty} d_{ji}(t-T_{ji}) = 0
\label{eq:8.100}
\end{eqnarray}
hold.
Thus, taking the limit of (\ref{eq:8.37}),
it follows
\begin{eqnarray}
\lim_{t\to \infty} (r^{\xi}_{ij} - r^{\xi}_{ij}(t-\bar T_{ij})) = 0.
\label{eq:8.34}
\end{eqnarray}
Subtracting (\ref{eq:8.33b}) at time $t-T_{ji}$ from (\ref{eq:8.33a}) yields
\begin{eqnarray}
r_{ij}^x  \!\!&\!\!+\!\!&\!\! r_{ij}^x(t-\bar T_{ij}) - 2 r_{ji}^x (t-T_{ji}) 
\nonumber\\
\!\!&\!\!=\!\!&\!\! (1/\eta)\{-a_{ij}d_{ij} 
+ a_{ij}d_{ji}(t-T_{ji}) + b_{ij}(r_{ij}^{\xi} - \xi_i)  
\nonumber\\
 \!\!&\!\!\!\!&\!\! \hspace{1cm}
-b_{ij} (r_{ij}^{\xi}(t-\bar T_{ij}) - \xi_i(t-\bar T_{ij}))\}.
\label{eq:8.35}
\end{eqnarray}
Since $\xi_i$ converges to a constant from (iii), we obtain
\begin{eqnarray}
\lim_{t\to \infty} (\xi_i - \xi_i(t-\bar T_{ij})) = 0.
\label{eq:8.36}
\end{eqnarray}
Taking the limit of (\ref{eq:8.35}) and using (\ref{eq:8.100}),
(\ref{eq:8.34}), and (\ref{eq:8.36}), we have
\begin{eqnarray}
\lim_{t\to \infty}(r_{ij}^x  + r_{ij}^x(t-\bar T_{ij}) - 2 r_{ji}^x (t-T_{ji})) = 0. 
\label{eq:8.38}
\end{eqnarray}
It is confirmed from (ii) and (\ref{eq:8.32a}) that
the $l$-th element of $\lim_{t\to \infty}(r_{ij}^x  + r_{ij}^x(t-\bar T_{ij}))$ in (\ref{eq:8.38}) is equal to $2 z^*_l$, which implies that
the $l$-th element of $r_{ji}^x$ converges to $z^*_l$.
This and (\ref{eq:8.32a}) also mean that the $l$-th element of $x_j$ converges to $z^*_l$.
Following the same procedure for a neighbor $k$ of $i$ or $j$, the $l$-th element of $x_k$ is proved to converge to $z^*_l$.
Repeating the same process, we can prove that the $l$-th element of $x_{j}$ converges to $z^*_l$ for all $j = 1,2,\dots, n$ 
because of Assumption \ref{ass:2}.
Convergence of the other elements is also proved in the same way.
\end{proof}

\section{Extension to Constrained Distributed Optimization}
\label{sec:4}

In this section, we consider the following constrained optimization problem.
\begin{eqnarray}
\min_{z\in \R^N} f(z) \mbox{ subject to } g_i(z) \leq 0\ {\forall i} = 1,2,\dots, n,
\label{eq:3.1}
\end{eqnarray}
where $f$ is defined in the same way as (\ref{eq:2.1}).
The functions $f_i$ and
$g_i: \R^N \to \R^{m_i}\ (i=1,2,\dots, n)$ are assumed to be 
private information of agent $i$ and 
the other agents do not have access to these functions.
In this section, we also assume that the optimal solution exists
and the minimal value of $f$ is finite.
Denoting the $l$-th element of $g_i$ by $g_{il} (l=1,2,\dots, m_i): \R^N \to \R$,
we assume the following assumptions.
\begin{assumption}
\label{ass:3}
The functions $f_i\ (i=1,2,\dots, n)$ are convex and twice differentiable.
The functions $g_{il}\ (l=1,2,\dots, m_i,\ i=1,2,\dots, n)$
are convex and continuously differentiable and their gradients, denoted by 
$\Gamma_i := \nabla g_i \in \R^{N \times m_i}\ (i=1,2,\dots, n)$,
are locally Lipschitz.
In addition, there exists $z$ such that $g_i(z) < 0 \ \ {\forall i}=1,2,\dots, n$.
\end{assumption}
\begin{assumption}
\label{ass:4}
The function $f$ is strictly convex.
\end{assumption}
Note that, if Assumptions \ref{ass:3} and \ref{ass:4} are satisfied,
the optimal solution $z^*$ to (\ref{eq:3.1}) is uniquely determined \cite{B_BK}.
It is well-known that, under these assumptions, 
$z^*$ is the optimal solution
to (\ref{eq:3.1}) if and only if there exist $\lambda_i^*\in \R^{m_i}\ (i=1,2,\dots, n)$ satisfying
the KKT condition \cite{B_BK}:
\begin{eqnarray}
\nabla f(z^*) + \sum_{i=1}^n  \Gamma_i(z^*) \lambda_i^* = 0,
\label{eq:3.2a}\\
\lambda_i^* \geq 0,\ g_i(z^*) \leq 0\ \ {\forall i}=1,2,\dots, n,
\label{eq:3.2b}\\
\lambda_{il}^* g_{il}(z^*) = 0 \ \ {\forall l}=1,2,\dots, m_i, \ {\forall i}=1,2,\dots, n,
\label{eq:3.2c}
\end{eqnarray}
where $\lambda_{il}^*$ is the $l$-th element of $\lambda_{i}^*$.
We finally define the Lagrangian $H$ for the problem (\ref{eq:3.1}) as
\begin{eqnarray}
H(z,\lambda) \!\!&\!\!:=\!\!&\!\! \sum_{i=1}^n H_i(z,\lambda_i),
\label{eq:3.3a}\\
H_i(z,\lambda_i) \!\!&\!\!:=\!\!&\!\! f_i(z) + \lambda_i^T g_i(z).
\nonumber
\end{eqnarray}



\subsection{PI Consensus-Based Distributed Optimization for Constrained Problem}
\label{sec:4.1}



Let us now define the dual function $\min_z H(z,\lambda)$ and $g(z) := [g_1^T(z)\ \cdots\ g_n^T(z)]^T \in \R^m\ (m := \sum_{i=1}^nm_i)$.
Then, inspired by the dual ascent method,
we consider the following dynamical system formulated based on \cite{arcak}.
\begin{eqnarray}
\dot \rho \!\!&\!\!=\!\!&\!\! 
\psi(\rho,g(\hat z^*(\rho))),\ \ \rho(0) \geq 0,
\label{eq:5.1}
\end{eqnarray}
where $\rho \in \R^{m}$, and 
\begin{eqnarray}
\hat z^*(\rho) := \arg\min_z H(z,\rho).
\label{eq:5.2}
\end{eqnarray}
Given vectors $g = [g_1^T\ \cdots\ g_n^T]^T \in \R^m$,  $\rho = [\rho^T_1\ \cdots\ \rho^T_n]^T\in \R^m$,
$g_i = [g_{i1}\ \cdots\ g_{im_i}]^T \in \R^{m_i}$, and $\rho_i = [\rho_{i1}\ \cdots\ \rho_{im_i}]^T$
$(i=1,2,\dots,n)$,
the functions
$\psi: \R^m\times \R^m \to \R^m$, 
$\psi_i: \R^{m_i}\times \R^{m_i} \to \R^{m_i}$
and $\psi_{il}: \R\times \R \to \R$ are defined as
\begin{eqnarray}
\hspace{-.5cm}&& \psi(\rho,g) := \begin{bmatrix}
\psi_{1}(\rho_{1},g_{1})\\
\vdots\\
\psi_{n}(\rho_{n},g_{n})
\end{bmatrix}
, \psi_i(\rho_i,g_{i})  := \begin{bmatrix}
\psi_{i1}(\rho_{i1},g_{i1})\\ 
\vdots\\ 
\psi_{im_i}(\rho_{im_i},g_{im_i})
\end{bmatrix}
\nonumber\\
\hspace{-.5cm}&&
\psi_{il}(\rho_{il},g_{il}) :=
\left\{
\begin{array}{ll}
0,&\mbox{if } \rho_{il} = 0\ \&\ g_{il} < 0 \\
g_{il},&\mbox{otherwise}
\end{array}
\right..
\label{eq:3.6}
\end{eqnarray}

Extracting the dynamics of $\rho_i$ from (\ref{eq:5.1}) yields
\begin{eqnarray}
\dot \rho_{i} = \psi_i(\rho_i,g_i(\hat z^*(\rho))),\ \rho_i(0) \geq 0,
\label{eq:3.5}
\end{eqnarray}
which does not depend on the information of other $g_j\ (j\neq i)$.
Thus, once $\hat z^*(\rho)$ is given,
the dynamics (\ref{eq:3.5}) can be run by agent $i$ 
using only the private constraint function $g_i$.

Let us next consider the calculation of (\ref{eq:5.2}).
We see from (\ref{eq:3.3a}) and (\ref{eq:5.2}) that,
once $\rho$ is fixed, the problem 
(\ref{eq:5.2}) to be solved here takes the same form as (\ref{eq:2.1})
except for the definition of the cost functions.
Inspired by the fact, we let each agent $i$ run the PI consensus algorithm
(\ref{eq:2.9a}) and (\ref{eq:2.9b}) with 
\begin{eqnarray}
u_i = -\alpha \nabla_{z} H_i(x_i,\rho_i) = -\alpha \phi_i(x_i) - \alpha  (\Gamma_i(x_i)) \rho_i.
\label{eq:3.4}
\end{eqnarray}
Remark that (\ref{eq:3.4}) relies only on the private functions $f_i$ and $g_i$, and local variables $\rho_i$ and $x_i$.
Since $x_i$ is regarded as an estimate of $\hat z^*(\rho)$, we replace $\hat z^*(\rho)$ in 
(\ref{eq:3.5}) by $x_i$ and reformulate the dynamics as
\begin{eqnarray}
\dot \rho_{i} = \psi_i(\rho_i,g_i(x_i)),\ \rho_i(0) \geq 0, 
\label{eq:5.3}
\end{eqnarray}
which also consists only of the local variables.
Note that the right-hand side of (\ref{eq:5.3}) includes discontinuity due to the definition of $\psi_{il}$ in (\ref{eq:3.6}),
but we let each agent run the dynamics so that the trajectories of $\rho_i$ are continuous, in other words,
jumps of $\rho_i$ are not artificially generated.
It is then immediately confirmed that if $\rho_i(0) \geq 0$ then $\rho_i(t) \geq 0$ for all subsequent time $t$
regardless of the trajectory of $g_i(x_i)$.



The collective dynamics of all agents is given by
\begin{eqnarray}
\begin{bmatrix}
\dot x\\
\dot \xi
\end{bmatrix}
\!\!&\!\!=\!\!&\!\! -
\begin{bmatrix}
\bLP&-\bLI\\
\bLI&0
\end{bmatrix}
\begin{bmatrix}
x\\
\xi
\end{bmatrix}
-\alpha
\begin{bmatrix}
\phi(x)+ \Gamma(x)\rho\\
0
\end{bmatrix},
\label{eq:5.4a}\\
\dot \rho \!\!&\!\!=\!\!&\!\!
\psi(\rho,\bar g(x)),
\label{eq:5.4b}
\end{eqnarray}
where  
$\Gamma(x)$ is the block diagonal matrix with diagonal blocks $\Gamma_1(x_1), \dots, \Gamma_n(x_n)$,
and $\bar g(x) := [g_1^T(x_1)\ \cdots\ g_n^T(x_n)]^T$.
The block diagram of the system is illustrated in Fig. \ref{fig:8}.
Notice that an outer feedback path is added to the solution of unconstrained problems in Fig. \ref{fig:2}.

\subsection{Equilibrium/Convergence Analysis: Constrained Case}
\label{sec:4.2}

\begin{figure}[t]
\centering
\includegraphics[width=8.4cm]{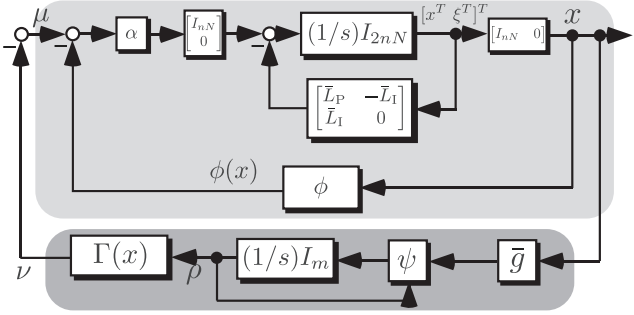}
\caption{Block diagram of the PI consensus-based distributed optimization algorithm
for constrained optimization.
The system colored by light gray is passive from $\tilde \mu = \mu - \mu^*$ to $\tilde x = x - x^*$ with $\mu^* := \Gamma(x^*)\lambda^*$
(Lemma \ref{lem:20}).
The system colored by dark gray 
is also passive from $\tilde x = x - x^*$ to $\tilde \nu = \nu - \nu^*$ with $\nu^* := \Gamma(x^*) \lambda^*$ (Lemma \ref{lem:5}).}
\label{fig:8}
\end{figure}

We consider Fig. \ref{fig:8} and cut the loop at the left of the block $\Gamma(x)$
and focus on the systems encircled by light and dark gray,
where the input of the former is denoted by $\mu$ and output of the latter is by $\nu$.
Then, these systems are formulated as
\begin{eqnarray}
\begin{bmatrix}
\dot x\\
\dot \xi
\end{bmatrix}
\!\!&\!\!=\!\!&\!\! -
\begin{bmatrix}
\bLP&-\bLI\\
\bLI&0
\end{bmatrix}
\begin{bmatrix}
x\\
\xi
\end{bmatrix}
-\alpha
\begin{bmatrix}
\phi(x)\\
0
\end{bmatrix}
+ \alpha 
\begin{bmatrix}
\mu\\
0
\end{bmatrix}
\label{eq:5.5a}
\end{eqnarray}
with output $x$ and
\begin{eqnarray}
\dot \rho \!\!&\!\!=\!\!&\!\!
\psi(\rho,\bar g(x)),\ \
\nu = \Gamma(x)\rho
\label{eq:5.6b}
\end{eqnarray}
respectively. Now, we obtain the following lemma.

\begin{lemma}
\label{lem:6}
Suppose that Assumptions \ref{ass:2}, \ref{ass:3} and \ref{ass:4} hold.
Denote the unique solution to (\ref{eq:3.2a})--(\ref{eq:3.2c}) by
$(z^*,\lambda^*)$, where $\lambda^* := [(\lambda_1^*)^T\ \cdots\ (\lambda_n^*)^T]^T\in \R^{m}$.
Then, there exists $\xi^*$ such that the pair of $x^* = (\textbf{1}_n\otimes z^*)$
and $\xi^*$ is an equilibrium of (\ref{eq:5.5a}) for the equilibrium input $\mu^* := \Gamma(x^*)\lambda^*$.
\end{lemma}
\begin{proof}
Consider the right-hand side of (\ref{eq:5.5a}).
Substituting $x = x^*$ yields
$\bLI x^* = (L_{\bf I}\textbf{1}_N)\otimes z^* = 0$ and hence $\dot \xi = 0$.
Replacing $\mu$ by $\mu^*$, we obtain 
\begin{eqnarray}
\dot x^*
=
\bLI \xi
- \alpha
(\phi(x^*) + \Gamma(x^*) \lambda^*).
\label{eq:3.20}
\end{eqnarray}
It is thus sufficient to prove that there exists $\xi$ such that
the right-hand side of (\ref{eq:3.20}) is zero.
In other words, we have only to prove that 
$\alpha(\phi(x^*) + \Gamma(x^*) \lambda^*)$ is included in the image of the matrix $\bLI$,
which is equivalent to
$({\bf 1}_n\otimes I_N)^T(\phi(x^*) + \Gamma(x^*) \lambda^*) = 0$.
This equation is immediately proved from (\ref{eq:3.2a}).
\end{proof}

Using the above lemma, we can prove the following result.

\begin{lemma}
\label{lem:20}
Suppose that Assumptions \ref{ass:2}, \ref{ass:3} and \ref{ass:4}  hold.
Then, the system (\ref{eq:5.5a})
is passive from $\tilde \mu := \mu - \mu^*$ to $\tilde x:= \xi - \xi^*$
with respect to the storage function $\frac{1}{\alpha}{\tilde{S}}$,
where $\tilde S$ is defined in (\ref{eq:2.18}).
\end{lemma}
\begin{proof}
Define $\tilde \xi := \xi - \xi^*$ for a fixed $\xi^*$.
Then, it follows from (\ref{eq:5.5a}) and Lemma \ref{lem:6} that
\begin{eqnarray}
\begin{bmatrix}
\dot{\tilde{x}}\\
\dot{\tilde{\xi}}
\end{bmatrix}
\!
= -
\begin{bmatrix}
\bLP&-\bLI\\
\bLI&0
\end{bmatrix}
\!
\begin{bmatrix}
\tilde x\\
\tilde \xi
\end{bmatrix}
\!
-\alpha
\!
\begin{bmatrix}
\phi(x) - \phi(x^*)\\
0
\end{bmatrix}
\!
+ \alpha 
\begin{bmatrix}
\tilde \mu\\
0
\end{bmatrix}.
\label{eq:5.7a}
\end{eqnarray}
Now, following the same procedure as Lemma \ref{lem:1},
we can immediately prove
\begin{eqnarray}
\hspace{-.5cm}\frac{1}{\alpha}\dot{\tilde{S}} \!\!&\!\!=\!\!&\!\! -\frac{1}{\alpha}\tilde x^T \bLP \tilde x - (x - x^*)^T(\phi(x)-\phi(x^*)) + \tilde \mu^T \tilde x
\label{eq:5.8}\\
\hspace{-.5cm}\!\!&\!\!\leq\!\!&\!\! \tilde \mu^T \tilde x.
\end{eqnarray}
This completes the proof.
\end{proof}

Let us next consider (\ref{eq:5.6b}).
Extracting $l$-th row of (\ref{eq:5.3}) yields
\begin{eqnarray}
\dot \rho_{il} = \psi_{il}(\rho_{il},g_{il}(x_{i})),\ \rho_{il}(0) \geq 0, 
\label{eq:12.3}
\end{eqnarray}
whose right-hand side can be discontinuous
at the states such that $\rho_{il}=0$ and $g_{il}(x_{i}) < 0$.
For convenience, the mode satisfying the upper condition in (\ref{eq:3.6})
is called mode 1 and the other is mode 2.
Then, the following lemma holds.

\begin{lemma}
\label{lem:5}
Suppose that Assumptions \ref{ass:3} and \ref{ass:4} hold.
Then, the system (\ref{eq:5.6b}) with $\rho(0) \geq 0$
is passive from $\tilde x$ to $\tilde \nu := \nu - \nu^*$ with $\nu^* := \Gamma(x^*) \lambda^*$
for the storage function
$\sum_{i=1}^n U_i$, where
$U_i := \frac{1}{2}\|\rho_i - \lambda_i^*\|^2$.
\end{lemma}
\begin{proof}
We first consider the time when no mode switch occurs.
The time derivative of $U_i$ along the system trajectories is then given by
\begin{eqnarray}
\dot U_i 
\!\!&\!\!=\!\!&\!\! \sum_{l=1}^{m_i} 
(\rho_{il} - \lambda_{il}^*)\psi_{il}(\rho_{il},g_{il}(x_i)).
\nonumber
\label{eq:3.9}
\end{eqnarray}
If mode 2 is active, $\psi_{il}(\rho_{il},g_{il}(x_i)) = g_{il}(x_i)$
and hence
\begin{eqnarray}
(\rho_{il} - \lambda_{il}^*)\psi_{il}(\rho_{il},g_{il}(x_i)) 
= (\rho_{il} - \lambda_{il}^*)g_{il}(x_i)
\label{eq:3.12}
\end{eqnarray}
holds.
If mode 1 is active, $\rho_{il}=0$ and $\psi_{il}(\rho_{il},g_{il}(x_i)) = 0$ hold, and hence we have
\begin{eqnarray}
\hspace{-1cm}&&(\rho_{il} - \lambda_{il}^*)\psi_{il}(\rho_{il},g_{il}(x_i)) = 0 = \rho_{il}g_{il}(x_i)
\nonumber\\
\hspace{-1cm}&& \hspace{2cm} = (\rho_{il} - \lambda_{il}^*)g_{il}(x_i) + \lambda_{il}^*g_{il}(x_i).
\nonumber
\label{eq:3.10}
\end{eqnarray}
Since $\lambda_{il}^* \geq 0$ from (\ref{eq:3.2b}) and $g_{il}(x_i) < 0$ from (\ref{eq:3.6}),
the term $\lambda_{il}^*g_{il}(x_i)$ is non-positive and hence we obtain
\begin{eqnarray}
(\rho_{il} - \lambda_{il}^*)\psi_{il}(\rho_{il},g_{il}(x_i)) 
\leq (\rho_{il} - \lambda_{il}^*)g_{il}(x_i).
\label{eq:3.11}
\end{eqnarray}

Let us next consider the time when a mode switch happens in (\ref{eq:5.3}) for some $i$
and $l$. 
In this case, $U_{il}(t) := \frac{1}{2}\|\rho_{il}(t) - \lambda_{il}^*\|^2$ can be indifferentiable in the standard sense.
We thus introduce the upper Dini derivative\footnote{The upper Dini derivative $D^+ U(t)$ of a scalar function $U$ is defined
as $D^+ U(t) = \lim\sup_{h \to 0+}\frac{U(t+h)-U(t)}{h}$} denoted by $D^+U_{il}$.
Then, it is given by either of $(\rho_{il} - \lambda_{il}^*)0 = 0$ or 
$(\rho_{il} - \lambda_{il}^*)g_{il}(x_i)$ depending on the sign of $(\rho_{il} - \lambda_{il}^*)$.
Thus, following the same procedure as above, we can confirm that
\begin{eqnarray}
D^+ U_i \leq \sum_{l=1}^{m_i}  (\rho_{il} - \lambda_{il}^*)g_{il}(x_i).
\label{eq:3.13}
\end{eqnarray}

From (\ref{eq:3.12}) and (\ref{eq:3.11}), the inequality (\ref{eq:3.13}) holds
for all time $t\in \R^+$.
This is rewritten as
\begin{eqnarray}
D^+ U_i  \leq (\rho_{i} - \lambda_{i}^*)^T\{g_{i}(x_i) - g_i(z^*)\}
+ (\rho_{i} - \lambda_{i}^*)^T g_i(z^*).
\label{eq:3.14}
\end{eqnarray}
Noticing that $\rho_{i}\geq 0$ and $g_i(z^*) \leq 0$ from (\ref{eq:3.2b}),
the inequality $\rho_{i}^Tg_i(z^*) \leq 0$ holds.
In addition, $(\lambda_{i}^*)^Tg_i(z^*) = 0$ is true from (\ref{eq:3.2c}).
Thus, (\ref{eq:3.14}) is further rewritten as
\begin{eqnarray}
D^+ U_i
\!\!&\!\!\leq\!\!&\!\! (\rho_{i} - \lambda_{i}^*)^T\{g_{i}(x_i) - g_i(z^*)\}
\nonumber\\
\!\!&\!\!=\!\!&\!\! 
\sum_{l=1}^{m_i} \left[\rho_{il}\{g_{il}(x_i) - g_{il}(z^*)\}
-\lambda_{il}^*
\{g_{il}(x_i) - g_{il}(z^*)\}\right].
\nonumber
\end{eqnarray}
Because of the convexity of $g_{il}$,
the following inequalities hold.
\begin{eqnarray}
g_{il}(x_i) - g_{il}(z^*) \geq (\nabla g_{il}(z^*))(x_i - z^*),
\nonumber\\
g_{il}(x_i) - g_{il}(z^*) \leq (\nabla g_{il}(x_i))(x_i - z^*).
\nonumber
\end{eqnarray}
Using these together with $\rho_i\geq 0$ and $\lambda_i^*\geq 0$, we can prove
\begin{eqnarray}
D^+ U_i
\!\!&\!\!\leq\!\!&\!\! 
\left\{\sum_{l=1}^{m_i}((\nabla g_{il}(x_i))\rho_{il} - (\nabla g_{il}(z^*))\lambda_{il}^*)\right\}^T(x_i - z^*)
\nonumber\\
\!\!&\!\!=\!\!&\!\!  \left\{(\Gamma_i(x_i))\rho_{i} - (\Gamma_i(z^*))\lambda_{i}^*\right\}^T(x_i - z^*).
\nonumber
\end{eqnarray}
Thus, we have
\begin{eqnarray}
\!\!\!\!\!\!\!\!
D^+  \left(\sum_{i=1}^n U_i\right)
\!\!&\!\!\leq\!\!&\!\! 
\{(\Gamma(x))\rho - (\Gamma(x^*))\lambda^*\}^T(x - x^*) 
\nonumber\\
\!\!&\!\!=\!\!&\!\! 
\tilde \nu^T\tilde x.
\label{eq:3.17}
\end{eqnarray}
Integrating (\ref{eq:3.17}) in time proves the original definition (\ref{eq:0.1}).
\end{proof}
Namely, the closed-loop system (\ref{eq:5.4a}) and (\ref{eq:5.4b})
is regarded as a feedback interconnection of two passive systems.

We are now ready to state the following convergence result.

\begin{theorem}
\label{thm:2}
Consider the system (\ref{eq:5.4a}) and (\ref{eq:5.4b}).
If Assumptions \ref{ass:2}, \ref{ass:3} and \ref{ass:4} hold,
then $x_i$ asymptotically converges to the optimal solution $z^*$ to
(\ref{eq:3.1}) for all $i = 1,2,\dots, n$.
\end{theorem}

\begin{proof}
Define $U := \frac{1}{\alpha}\tilde S + \sum_{i=1}^n U_i$.
From (\ref{eq:5.8}) and (\ref{eq:3.17}), 
\begin{eqnarray}
D^+ U \!\!&\!\!=\!\!&\!\!  \frac{1}{\alpha}\dot{\tilde{S}} + 
D^+ \left(\sum_{i=1}^n U_i\right)
\nonumber\\
\!\!&\!\!\leq\!\!&\!\! 
- \frac{1}{\alpha}x^T\bLP x 
- (x-x^*)^T(\phi(x)-\phi(x^*)) \leq 0
\nonumber
\end{eqnarray}
holds. Integrating this in time, we have
\begin{eqnarray}
\int_{0}^{\infty}\Big(
\frac{1}{\alpha}x^T\bLP x +
(x-x^*)^T(\phi(x)-\phi(x^*)\Big) dt < \infty.
\end{eqnarray}
We also see $x, \xi, \rho \in {\mathcal L}_{\infty}$ since $U$ is positive definite.
Thus, from (\ref{eq:5.4a}), $\dot x \in {\mathcal L}_{\infty}$ also holds.
The time derivative of $\frac{1}{\alpha}x^T\bLP x +
(x-x^*)^T(\phi(x)-\phi(x^*))$ is given as
\begin{eqnarray}
\ \ 
\frac{2}{\alpha} x^T\bLP\dot x + 
(\phi(x)-\phi(x^*))^T \dot x + (x - x^*)^T(\nabla \phi(x))\dot x,
\nonumber
\end{eqnarray}
which is also bounded.
Thus, invoking Barbalat's lemma,
we can prove $\lim_{t\to \infty}x^T\bLP x = 0$ and $\lim_{t\to \infty}(x-x^*)^T(\phi(x)-\phi(x^*)) = 0$.
The first equation means that there exists a trajectory $c(\cdot)$ such that
$\lim_{t\to \infty}(x_i - c) = 0$ for all $i$.
The second equation is then rewritten as 
\begin{eqnarray}
\!\!&\!\!\!\!&\!\!\sum_{i=1}^n\lim_{t\to \infty} (x_i - z^*)^T (\phi_i(x_i) - \phi_i(z^*))
\nonumber\\
\!\!&\!\!=\!\!&\!\!  \sum_{i=1}^n \lim_{t\to \infty}((c - z^*)^T (\phi_i(c) - \phi_i(z^*))
+\sigma(x_i) - \sigma(c)) = 0,
\nonumber
\end{eqnarray}
where $\sigma(x_i) := x_i^T\phi_i(x_i) - x_i^T\phi_i(z^*)-(z^*)^T\phi_i(x_i) $.
Since $\sigma$ is continuous and $\lim_{t\to \infty}(x_i - c) = 0$,
we have $\lim_{t\to \infty}(\sigma(x_i) - \sigma(c)) = 0$ and hence
$\sum_{i=1}^n \lim_{t\to \infty}(c - z^*)^T (\phi_i(c) - \phi_i(z^*)) = 0$.
From this equation, we can also prove $c \to z^*$ in the same way as 
asymptotic stability in Lyapunov theorem.
It is thus concluded that $x_i \to z^*$ hold for all $i$.
\end{proof}

The present results can be easily extended to the problem with equality constraints using
techniques e.g. in \cite{YT_TR}, but we omit including the result to reduce complexity of the paper.
Inequality constraints for another class of distributed optimization are
also addressed in \cite{FP_AT10} using another energy functions,
where the authors employ LaSalle's principle for hybrid systems \cite{sastry}.
The completed version of the proof of the result in \cite{FP_AT10} is presented in \cite{cortes}.
Besides the difference of the problem formulation, a more explicit contribution
relative to \cite{FP_AT10,cortes} is presented in the next subsection.

\subsection{Constrained Distributed Optimization with Delay}
\label{sec:4.3}

In this subsection, we extend the results in Section \ref{sec:3} to
 the constrained problem (\ref{eq:3.1}).

Let us revisit the agent $i$'s dynamics (\ref{eq:2.9a}), (\ref{eq:2.9b}), (\ref{eq:3.4})
and (\ref{eq:5.3}).
Similarly to Section \ref{sec:3},
$x_j$ and $\xi_j$ are not directly received from agent $j \in \N_i$ in the presence of delays
and we again denote the received information as $r^x_{ij}$ and $r^{\xi}_{ij}$, respectively.
Then, the system is formulated as
\begin{eqnarray}
\dot x_i \!\!&\!\!=\!\!&\!\! \sum_{j\in \N_i}a_{ij}(r^x_{ij} - x_i) - \sum_{j\in \N_i}b_{ij}(r^{\xi}_{ij} - \xi_i) 
\nonumber\\
\!\!&\!\!\!\!&\!\! \hspace{2cm} -  \alpha (\phi_i(x_i)+(\Gamma_i(x_i))\rho_i),
\label{eq:7.1a}\\
\dot \xi_i \!\!&\!\!=\!\!&\!\! \sum_{j\in \N_i}b_{ij}(r^x_{ij} - x_i),
\label{eq:7.1b}\\
\dot \rho_{i} \!\!&\!\!=\!\!&\!\! \psi_i(\rho_i,g_i(x_i)),\ \rho_i(0) \geq 0,
\label{eq:7.1c}
\end{eqnarray}
whose block diagram is illustrated in Fig. \ref{fig:9}.

\begin{figure}
\centering
\includegraphics[width=8.4cm]{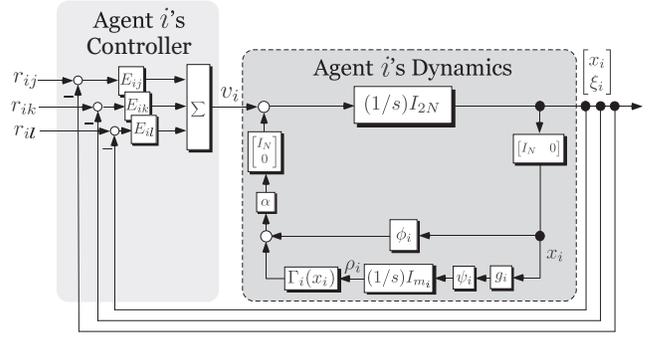}
\caption{Block diagram of agent $i$'s dynamics for constrained optimization 
in the presence of communication delays.
The system enclosed by the dashed line is passive 
from $\bar v_i$ to $[\bar x_i^T\ 
\bar \xi^T_i]^T$ (Lemma \ref{lem:9}), where $\bar v_i$, $\bar x_i$ and $\bar \xi_i$
are defined in (\ref{eq:7.2a}) and (\ref{eq:8.12a}).}
\label{fig:9}
\end{figure}


Now, we focus on the blocks encircled by the dashed line in Fig. \ref{fig:9} whose
system formulation is given by
\begin{eqnarray}
\begin{bmatrix}
\dot{{x}}_i\\
\dot{{\xi}}_i 
\end{bmatrix}
= v_i
-  \alpha 
\begin{bmatrix}
\phi_i(x_i) + \Gamma_i(x_i)\rho_i\\
0
\end{bmatrix}
\label{eq:7.4d}.
\end{eqnarray}
Then, we have the following lemma.

\begin{lemma}
\label{lem:9}
Suppose that Assumptions \ref{ass:2}, \ref{ass:11}, and \ref{ass:3} hold.
Then, the system (\ref{eq:7.4d}) is passive from 
$\bar v_i$ to $[\bar x_i^T\ \bar \xi_i^T]^T$
with respect to the storage function $W_i := \bar S_i + \alpha U_i$,
where $\bar S_i$ and $U_i$ are defined in 
Lemmas \ref{lem:18} and \ref{lem:5}, respectively, and
$\bar x_i$, $\bar \xi_i$ and $\bar v_i$ are defined in the same way as
 (\ref{eq:7.2a}) and (\ref{eq:8.12a}).
\end{lemma}
\begin{proof}
From Lemma \ref{lem:6} and (\ref{eq:7.4d}), we have
\begin{eqnarray}
\begin{bmatrix}
\dot{\bar{x}}_i\\
\dot{\bar{\xi}}_i 
\end{bmatrix}
= \bar v_i
-  \alpha 
\begin{bmatrix}
\phi_i(x_i) - \phi_i(z^*) + \Gamma_i(x_i)\rho_i
- \Gamma_i(z^*)\lambda^*_i\\
0
\end{bmatrix}
\label{eq:7.4a}
\end{eqnarray}
Similarly to the proof of Lemma \ref{lem:18},
it follows
\begin{eqnarray}
\dot{\bar{S}}_i \!\!&\!\!=\!\!&\!\! \begin{bmatrix}
{\bar{x}}_i\\
{\bar{\xi}}_i
\end{bmatrix}^T \bar v_i - \alpha(x_i - z^*)^T\{\phi_i(x_i)-\phi_i(z^*) 
\nonumber\\
\!\!&\!\!\!\!&\!\!
\hspace{1.5cm}
+ (\Gamma_i(x_i))\rho_i - (\Gamma_i(z^*))\lambda^*_i\}.
\label{eq:7.6}
\end{eqnarray}
We also see from Lemma \ref{lem:5} that
\begin{eqnarray}
\alpha D^+ \dot{U}_i \leq \alpha(x_i - z^*)^T\{(\Gamma_i(x_i))\rho_i - (\Gamma_i(z^*))\lambda^*_i\}.
\label{eq:7.7}
\end{eqnarray}
Combining (\ref{eq:7.6}) and (\ref{eq:7.7}), we obtain
\begin{eqnarray}
D^+  W_i \leq
\begin{bmatrix}
{\bar{x}}_i\\
{\bar{\xi}}_i
\end{bmatrix}^T \bar v_i - \alpha(x_i - z^*)^T(\phi_i(x_i)-\phi_i(z^*)).
\label{eq:7.8}
\end{eqnarray}
Integrating this in time proves the lemma.
\end{proof}

The above lemma means that the agent's passive dynamics in Fig. \ref{fig:3} is just
replaced by another passive dynamics with the same input-output pair.
Thus, we take the same inter-agent communication strategy as Section \ref{sec:3}.
Then, (\ref{eq:8.16}) again holds true and hence we have the following result.

\begin{theorem}
\label{thm:9}
Consider the system (\ref{eq:7.1a}), (\ref{eq:7.1b}) and (\ref{eq:7.1c}) for all $i$ with the scattering 
transformation (\ref{eq:8.13a}) and (\ref{eq:8.13b}) and delays (\ref{eq:8.21}) for all $j\in \N_i$ and all $i$.
If Assumptions \ref{ass:2}, \ref{ass:5} and \ref{ass:3} hold, 
then $x_i$ asymptotically converges to the optimal solution $z^*$ to (\ref{eq:3.1})
for all $i=1,2,\dots, n$.
\end{theorem}
\begin{proof}
Define the energy function
\[
W := \sum_{i=1}^n W_i + \sum_{(i,j)\in \E}V_{ij}.
\]
Then, combining (\ref{eq:7.8}) and (\ref{eq:8.16}), we obtain
\begin{eqnarray}
D^+  W
\!\!&\!\!\leq\!\!&\!\! - \sum_{i=1}^n\sum_{j\in \N_i}a_{ij}\|{\bar{x}}_i - \bar r^x_{ij}\|^2 
\nonumber\\
\!\!&\!\!\!\!&\!\!
 - \alpha\sum_{i=1}^n (x_i - z^*)^T(\phi_i(x_i)-\phi_i(z^*)) \leq 0,
\label{eq:7.17}
\end{eqnarray}
which implies that $x_i, \xi_i \in {\mathcal L}_{\infty}\ {\forall i}$.

Using (\ref{eq:8.11}), (\ref{eq:8.13a}), (\ref{eq:8.13b}), and (\ref{eq:8.21}),
we can derive 
\begin{eqnarray}
r_{ij}(t) \!\!&\!\!=\!\!&\!\! \bar E^2_{ij} r_{ij}(t-T_{ij}-T_{ji}) + \beta_{ij}(t)
\label{eq:8.22a}\\
r_{ji}(t) \!\!&\!\!=\!\!&\!\! \bar E^2_{ij} r_{ji}(t-T_{ij}-T_{ji}) + \beta_{ji}(t)
\label{eq:8.22b}
\end{eqnarray}
by calculation, where $\bar E_{ij} := (E_{ij}+\eta I_{2N})^{-1}(E_{ij} - \eta I_{2N})$, and 
$\beta_{ij}$ and $\beta_{ji}$
are linear functions of the states $[x_i^T\ \xi_i]^T$ and $[x_j^T\ \xi_j]^T$
at times $t$, $t-T_{ij}$, $t-T_{ji}$ and $t-T_{ij}-T_{ji}$.
Remark that $\beta_{ij}$ and $\beta_{ji}$ are both 
bounded since $x_i, \xi_i \in {\mathcal L}_{\infty}\ {\forall i}$.
It is also confirmed by calculation that all the eigenvalues of $\bar E_{ij}^2$ 
lie within the unit circle for any $\eta > 0, a_{ij}$ and $b_{ij}$.
Thus, both of (\ref{eq:8.22a}) and (\ref{eq:8.22b}) are stable difference equations with
bounded inputs and hence $r_{ij}, r_{ji}\in {\mathcal L}_{\infty}$.
Therefore, $\dot x_i \in {\mathcal L}_{\infty}\ {\forall i}$.

Integrating both sides of (\ref{eq:8.17}) proves that
\begin{eqnarray}
\int_0^{\infty}  (x_i - z^*)^T(\phi_i(x_i)-\phi_i(z^*)) < \infty
\nonumber
\end{eqnarray}
The derivative of $ (x_i - z^*)(\phi_i(x_i)-\phi_i(z^*))$ is given as
\begin{eqnarray}
(\phi_i(x_i)-\phi_i(z^*))^T \dot x_i + (x_i - z^*)^T\nabla^2 f_i(x_i)\dot x_i,
\nonumber
\end{eqnarray}
which is bounded. Thus, using Barbalat's lemma, we can prove
\begin{eqnarray}
\lim_{t\to \infty}(x_i - z^*)^T(\phi_i(x_i)-\phi_i(z^*)) = 0 \ {\forall i}=1,\dots, n.
\label{eq:8.25}
\end{eqnarray}
Under Assumption \ref{ass:5} and \ref{ass:3}, it is equivalent to
$x_i \to z^*$.
\end{proof}

Assumption \ref{ass:5} can be relaxed to \ref{ass:11} if we add the following
stronger assumption on the communication.
To state the assumption, we define a new graph $G' = (\{1,2,\dots, n\}, \E')$
such that $(j,k) \in \E'$ iff there exists $i$ satisfying $j \in \N_i$ and $k \in \N_i$.
Then, we assume the following.
\begin{assumption}
\label{ass:12}
The delays are homogeneous, namely $T_{ij} = T\ {\forall i,j}$ for some $T$.
In addition, the graph $G'$ is connected.
\end{assumption}

\begin{theorem}
\label{thm:20}
Consider the same system as Theorem \ref{thm:9}.
If Assumptions \ref{ass:11}, \ref{ass:3} and \ref{ass:12} hold and $b_{ij}$ is common, i.e.,
$b_{ij} = b\ {\forall j}\in \N_i,\ {\forall i}$ holds for some $b$, 
then $x_i$ asymptotically converges to the optimal solution $z^*$ to (\ref{eq:3.1})
for all $i=1,2,\dots, n$.
\end{theorem}

\begin{proof}
The inequality (\ref{eq:7.17}) means that $(x_i - r^x_{ij}) \in {\mathcal L}_2$
for all $j\in \N_i$ and $i = 1,2,\dots, n$.
As proved in Theorem \ref{thm:9}, $\dot x_i \in {\mathcal L}_{\infty}\ {\forall i}$,
and we can prove $\dot \xi_i \in {\mathcal L}_{\infty}\ {\forall i}$ in the same way.
Solving (\ref{eq:8.22a}), $r_{ij}(t)$ is given by a convolution sum of the input $\beta_{ij}$
whose time derivative is bounded.
Thus, we have
$\dot r_{ij} \in {\mathcal L}_{\infty}\ {\forall i}$ for all $j\in \N_i$ and $i = 1,2,\dots, n$.
Invoking Barbalat's lemma, we obtain 
\begin{eqnarray}
\lim_{t\to \infty} (x_i - r^x_{ij}) = 0 \ \ {\forall j}\in \N_i \mbox{ and } {\forall i} = 1,2,\dots, n.
\label{eq:15.4}
\end{eqnarray}

Under Assumption \ref{ass:12}, there exists $i$ such that $|\N_i| \geq 2$.
Take two neighbors $j, k\in \N_i$ of such $i$.
Then, (\ref{eq:15.4}) implies that
\begin{eqnarray}
\lim_{t\to \infty} (r_{ij}^x - r_{ik}^x) = 0.
\label{eq:15.5}
\end{eqnarray}
Following the same procedure as Theorem \ref{thm:10}, the equations
(\ref{eq:8.34}) and (\ref{eq:8.35}) hold.
From (\ref{eq:15.4}), we can also prove (\ref{eq:8.100}).
Remark that (\ref{eq:8.36})  is not proved here.
Taking the limit of (\ref{eq:8.35}) with $T_{ij} = T_{ji} = T$ and $b_{ij}= b$ and using (\ref{eq:8.100}) and (\ref{eq:8.34})  yield
\begin{eqnarray}
\lim_{t\to \infty} \{r_{ij}^x  \!\!&\!\!+\!\!&\!\!  r_{ij}^x(t-2T) -
2 r_{ji}^x(t-T) 
\nonumber\\
&&\hspace{1cm}+ (b/\eta)(
\xi_i - \xi_i(t-2T))\} = 0
\label{eq:15.1}
\end{eqnarray}
under Assumption \ref{ass:12}.
The same equation holds for $k$ as
\begin{eqnarray}
\lim_{t\to \infty} \{r_{ik}^x  \!\!&\!\!+\!\!&\!\! r_{ik}^x(t-2T) -
2 r_{ki}^x (t-T) 
\nonumber\\
&&\hspace{1cm}
+ (b/\eta)(\xi_i - \xi_i(t-2T))\} = 0.
\label{eq:15.2}
\end{eqnarray}
Subtracting  (\ref{eq:15.2}) from  (\ref{eq:15.1}) and using (\ref{eq:15.5}),
we have
\begin{eqnarray}
\lim_{t\to \infty} (r_{ji}^x - r_{ki}^x) = 0.
\label{eq:15.3}
\end{eqnarray}
From (\ref{eq:15.4}),
 (\ref{eq:15.3}) means that  $\lim_{t \to \infty}(x_j - x_k) = 0$,
which holds for any pair such that $(j,k)\in \E'$. 
Under Assumption \ref{ass:12}, we can conclude that
there exists a trajectory $c(\cdot)$ such that
$\lim_{t\to \infty}(x_i - c) = 0$ for all $i$.
The remaining of the proof is the same as Theorem \ref{thm:2}.
\end{proof}

\section{Application to Human Localization Using Pedestrian Detection Algorithm}

\begin{figure}[t]
\centering
\includegraphics[width=8.4cm]{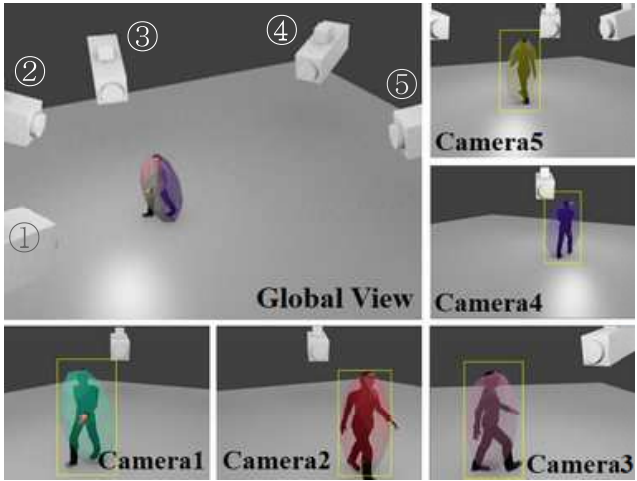}
\caption{Intended scenario of 3-D visual human localization.
Each camera acquires a rectangle enclosing the human (yellow boxes in the small windows)
using a pedestrian detection algorithm. 
The transparent colored ellipsoids in the small windows are final estimates generated by 
the algorithm presented in Subsection \ref{sec:4.3}, and all of them are also shown in the large window.}
\label{fig:14}
\end{figure}


We finally apply the proposed algorithms 
to the visual human localization problem investigated in \cite{Nec}.
Here, multiple networked cameras are assumed to be distributed over the 3-D Euclidean 
space to monitor a human as shown in Fig. \ref{fig:14}.
Each camera acquires 2-D rectangles on its own image plane in which the human lives, as shown in the small windows of Fig. \ref{fig:14},
by executing a pedestrian detection algorithm e.g. in \cite{DT_CVPR}.
Then, if camera $i$ detects the human, then it knows that 
the human must be inside of a cone ${\mathcal H}_i$ defined by connecting 
the focal center and the vertices of the rectangle.
In this paper, we suppose that all of the five cameras detect the human. 
Please refer to \cite{Nec} for the case where some cameras do not detect.

\begin{figure}[t]
\centering
\begin{minipage}{4.2cm}
\includegraphics[width=4.2cm]{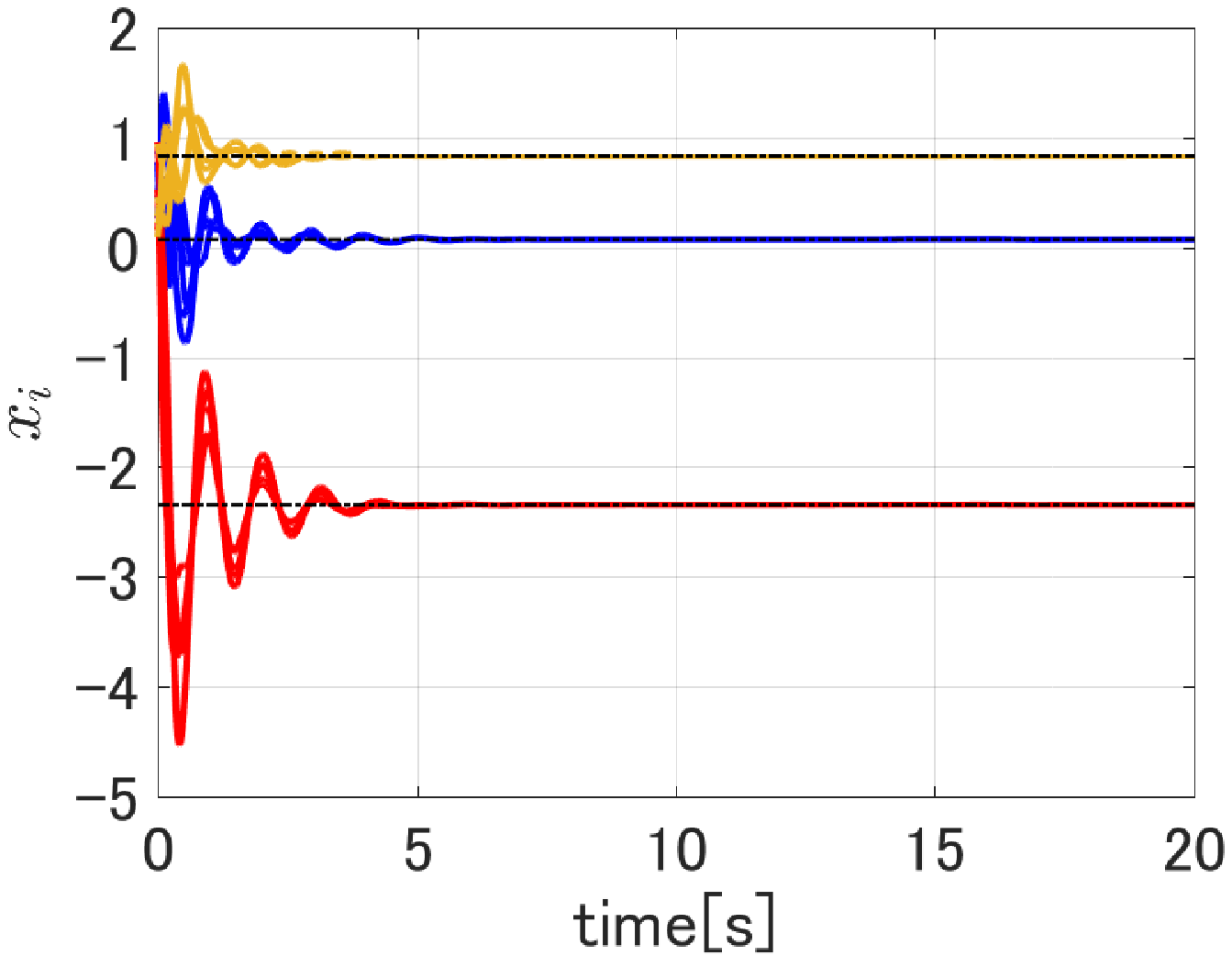}
\end{minipage}
\begin{minipage}{4.2cm}
\includegraphics[width=4.2cm]{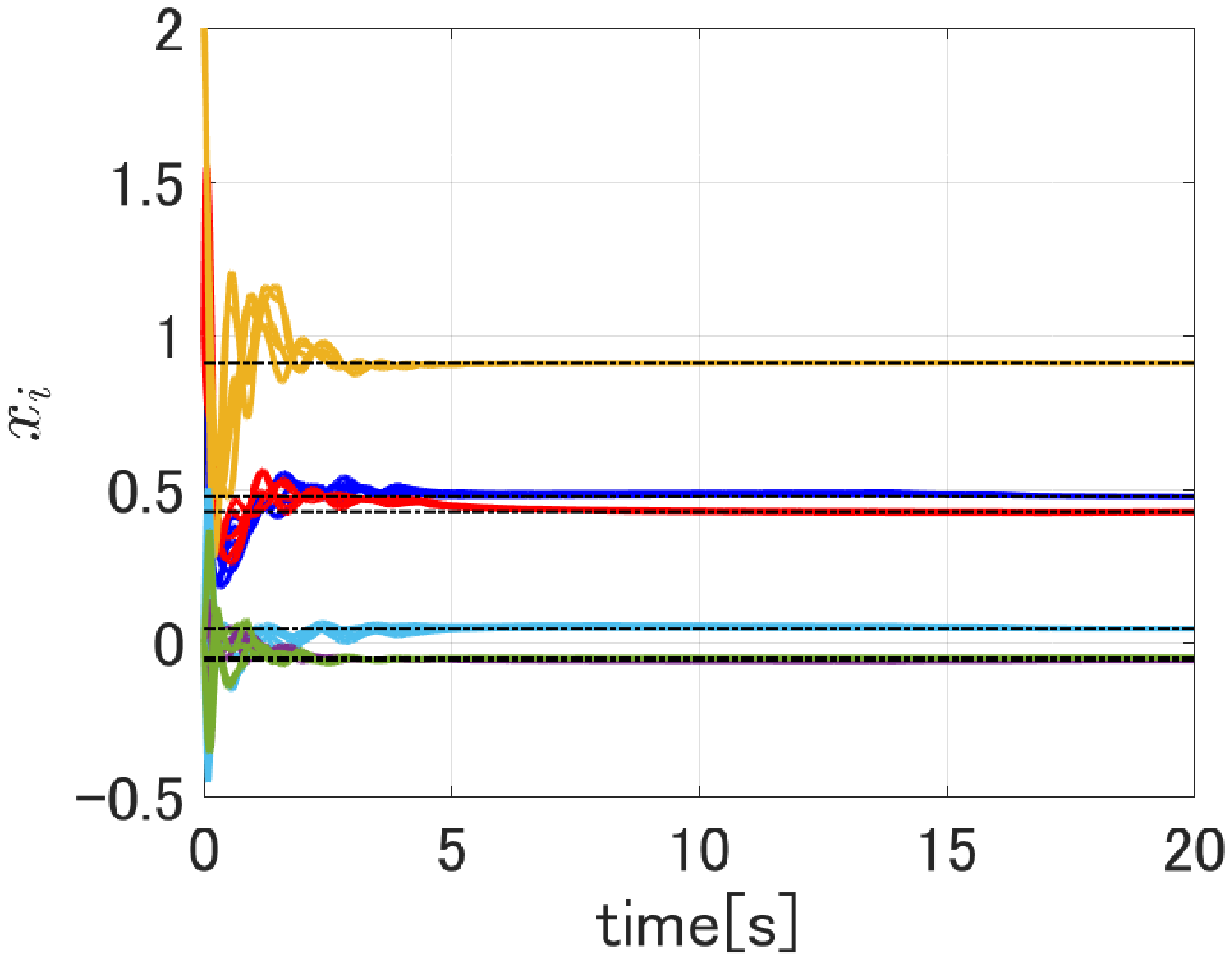}
\end{minipage}
\caption{Trajectories of $x_1, \dots, x_5$ without communication delays (left: vector variable $q$, right: matrix variable $Q$).}
\label{fig:10}
\medskip

\begin{minipage}{4.2cm}
\includegraphics[width=4.2cm]{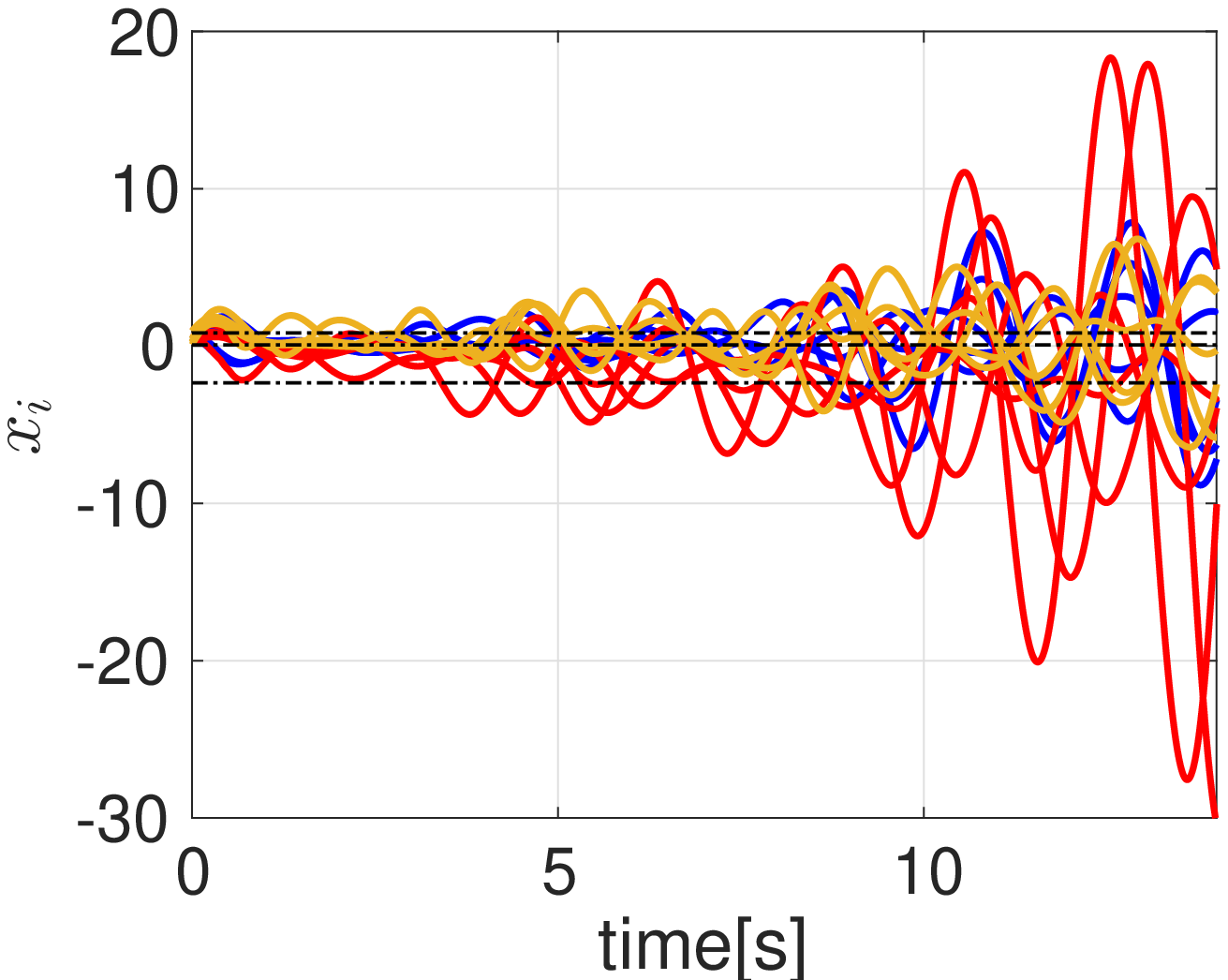}
\end{minipage}
\begin{minipage}{4.2cm}
\includegraphics[width=4.2cm]{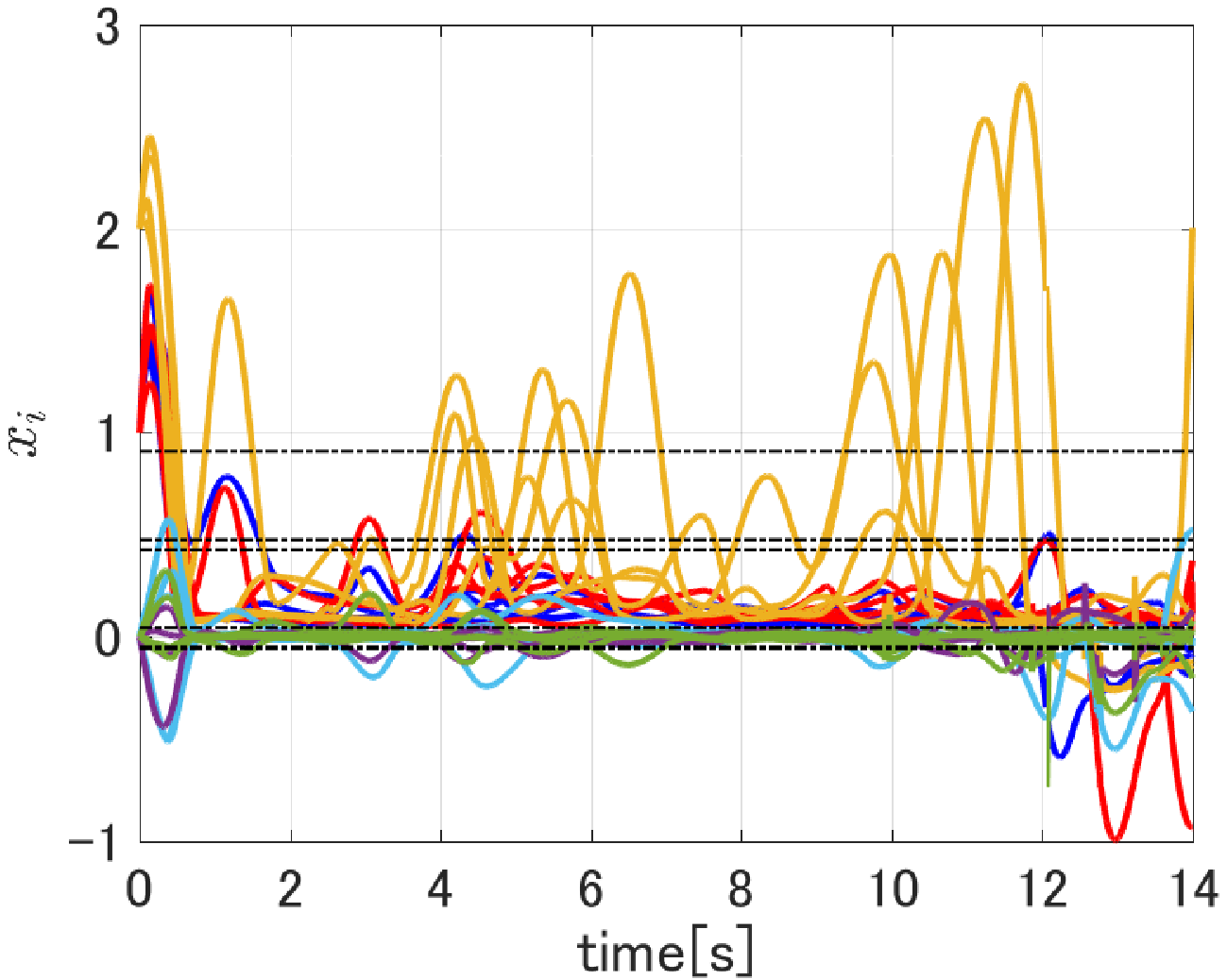}
\end{minipage}
\caption{Trajectories of $x_1, \dots, x_5$ with communication delays and without the scattering transformation (left: vector variable $q$, right: matrix variable $Q$).}
\label{fig:11}
\medskip

\begin{minipage}{4.2cm}
\includegraphics[width=4.2cm]{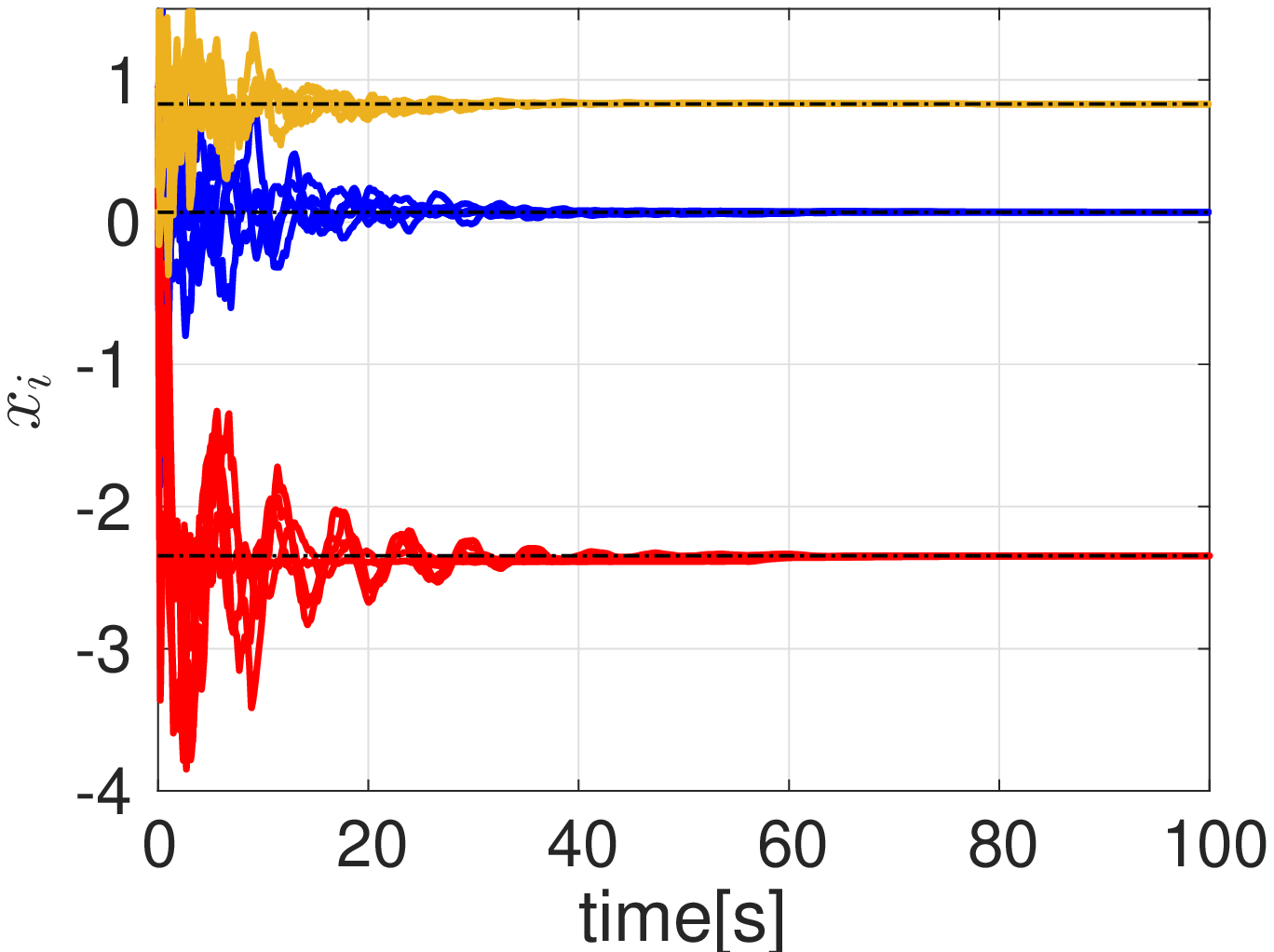}
\end{minipage}
\begin{minipage}{4.2cm}
\includegraphics[width=4.2cm]{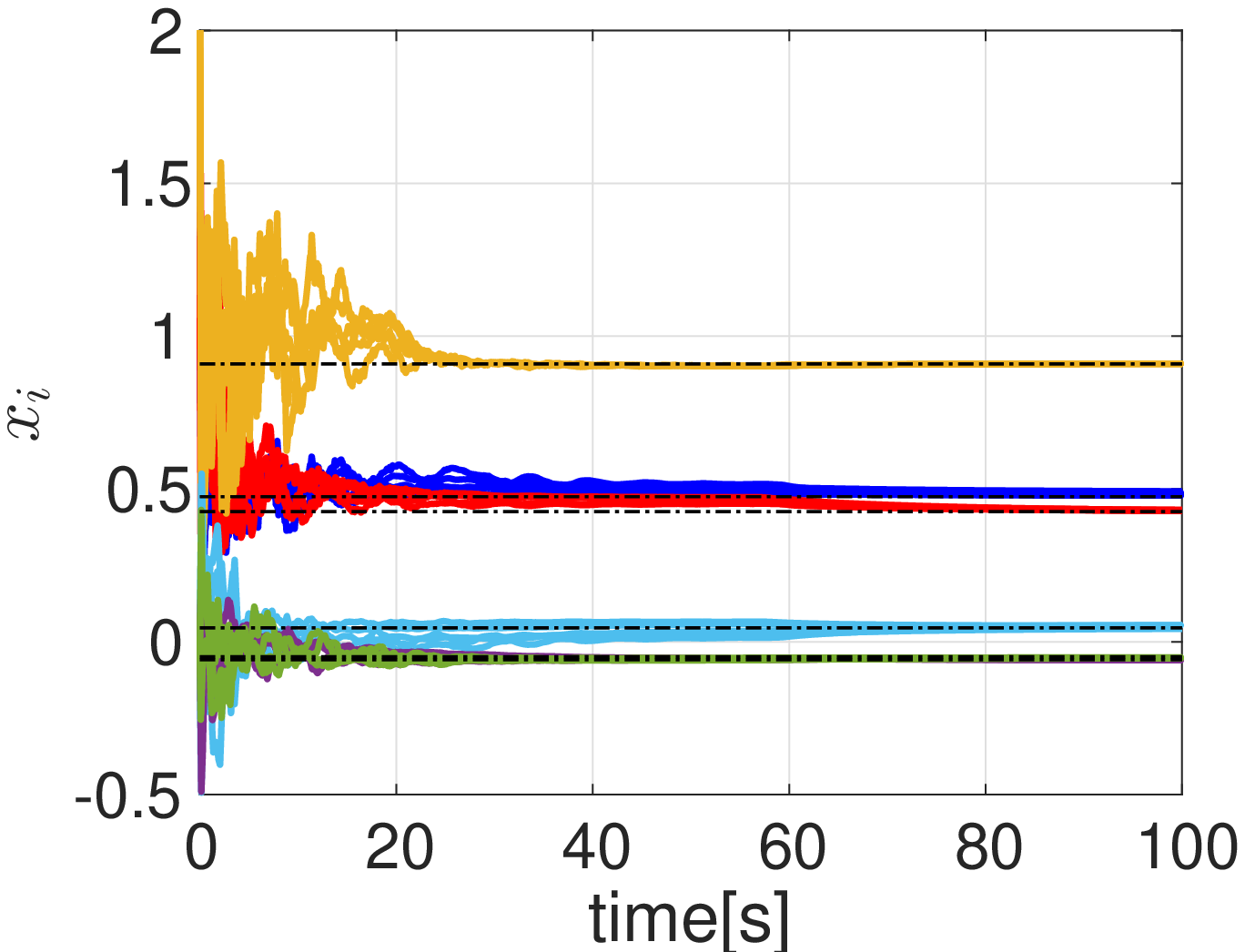}
\end{minipage}
\caption{Trajectories of $x_1, \dots, x_5$  with communication delays and scattering transformation 
(left: vector variable $q$, right: matrix variable $Q$).}
\label{fig:12}
\end{figure}

If the human is modeled as an ellipsoid 
$\Omega(q,Q) = \{p \in \R^3|\ (p-q)^TQ^{-2}(p-q)\leq 1\}$,
the decision variable $z$ consists of the elements of 
$q\in \R^3$ and $Q \in {\mathbb S}^{3\times 3}$.
Note that the symmetric matrices are parametrized by 6 variables.
Then, \cite{Nec} formulates the local cost function
\[
f_i(z) = -\log\det(Q) + w\min_{p \in {\mathcal C}_i}\|p - q\|^2,
\] 
where ${\mathcal C}_i$ is the line segment connecting the focal center and
the center of the rectangle on the image.
The scalar $w > 0$ is a weighting coefficient, and it is set to 
$w = 1$ in this simulation.
The local constraints are given by
$\Omega(q,Q) \subseteq {\mathcal H}_i$ and $Q > 0$, which
are reduced to the form of $g_i(z)\leq 0$.
See \cite{Nec} for more details on the problem formulation.
Note that the problem 
satisfies Assumptions \ref{ass:3}
and \ref{ass:4} in a realistic situation, 
but does not Assumption \ref{ass:11}.
It is actually confirmed that the present algorithm works
for the problem, but, to ensure preciseness, we add a term
$10^{-5} \|q\|^2$ to $f_i$.

We first run the algorithm presented in Subsection \ref{sec:4.1}
without adding communication delays.
The communication network is set to a ring graph, where
$a_{ij}$ and $b_{ij}$ are selected as $a_{ij}=1$ 
and $b_{ij} = 3$ for all $(i,j)\in {\mathcal E}$.
The initial values of the estimates of $q$ are randomly selected within $[0\ 1]$ and 
those for $Q$ are set to a diagonal matrix with elements $1$, $1$, and $2$ for all $i$\footnote{We can prove that if the initial estimate of $Q$ is positive definite,
they remain positive definite, over which the optimization problem is ensured to be convex, 
for all subsequent time. See Lemma 5 in \cite{Nec} for more details.}.
The initial values of $\rho_i$ are also randomly selected within $[0\ 1]$, and 
$\xi_i(0) = [1\ 2\ 3\ 1\ 1\ 2\ 0\ 0\ 0]^T$
for all $i$.
The gain $\alpha$ is set to $\alpha = 2$.
Then, the trajectories of the estimates $x_1, x_2, \dots, x_5$
are illustrated in Fig. \ref{fig:10},
where the dashed line describes the actual optimal solution.
We see that the estimates converge to the solution.

Let us next add communication delays, where each
$T_{ij}$ is selected randomly within $[0\ 1]$.
Then, we run the above algorithm under the same setting.
In the presence of delays, the trajectories of $x_1, x_2, \dots, x_5$
are changed to Fig. \ref{fig:11}, namely they diverge and the simulation stops
with errors.

We finally implement the algorithm presented in Subsection \ref{sec:4.3}.
Then, the resulting trajectories $x_1, x_2, \dots, x_5$
are shown in Fig. \ref{fig:12}.
We see that the system is stabilized by the scattering transformation,
and they successfully converge the optimal solution.
The final estimates of the ellipsoid are illustrated in Fig. \ref{fig:14},
where we see that every camera successfully computes an ellipsoid
tightly enclosing the human.

The problem of slow convergence is left as a future work.
Application of the technique in \cite{dorfler} can be an option.


\section{Conclusion}

In this paper, we have addressed a class of distributed optimization problems
in the presence of the inter-agent communication delays.
To this end, we first have focused on unconstrained distributed optimization problem, and 
presented a passivity-based perspective for the PI consensus-based
distributed optimization algorithm.
We then have proved that the inter-agent communication delays can be integrated 
while ensuring the convergence property using scattering transformation.
Moreover, we have extended the results to distributed optimization with local inequality 
constraints, and presented a passivity-based solution both in the absence and presence of delays.
Finally, the present algorithm has been applied to a visual human localization problem.




\end{document}